\documentclass[twocolumn]{article}
\usepackage{imr}

\usepackage{graphicx}
\usepackage{amsthm}
\usepackage{amsmath}
\usepackage{amssymb}
\usepackage{mathtools}
\usepackage{url}
\usepackage{cleveref}
\usepackage{algorithm}
\usepackage{algorithmic}
\usepackage{float}
\usepackage{placeins}

\DeclareMathOperator*{\Index}{I}
\newcommand{\R}{\mathbb R}

\newcounter{theorem}
\theoremstyle{theorem}
\newtheorem{proposition}[theorem]{Proposition}

\graphicspath{ {partition-simplification-figures/} }

\begin{document}

\title{\uppercase{Coarse Quad Layouts through Robust Simplification of Cross Field Separatrix Partitions}}
\author{Ryan Viertel$^1$ \and Braxton Osting$^2$ \and 
	Matthew Staten$^1$}
\date{
  $^1$Sandia National Laboratories, Albuquerque, NM, U.S.A. rvierte@sandia.gov, mlstate@sandia.gov\\
$^2$University of Utah, Salt Lake City, UT, U.S.A. osting@math.utah.edu
}

% \abstract{*} and \keywords{*} must be before \maketitle.
\abstract{
  Streamline-based quad meshing algorithms use smooth cross fields to partition surfaces into quadrilateral regions by tracing cross field separatrices. In practice, re-entrant corners and misalignment of singularities lead to small regions and limit cycles, negating some of the benefits a quad layout can provide in quad meshing. We introduce three novel methods to improve on a pipeline for coarse quad partitioning. First, we formulate an efficient method to compute high-quality cross fields on curved surfaces by extending the diffusion generated method from Viertel and Osting (SISC, 2019). Next, we introduce a method for accurately computing the trajectory of streamlines through singular triangles that prevents tangential crossings. Finally, we introduce a robust method to produce coarse quad layouts by simplifying the partitions obtained via naive separatrix tracing. Our methods are tested on a database of 100 objects and the results are analyzed. The algorithm performs well both in terms of efficiency and visual results on the database when compared to state-of-the-art methods.
}


\maketitle
\thispagestyle{empty}
\pagestyle{empty}

\section{Introduction} \label{sec:intro}

Block structured quad meshes are often desirable because of their numerical efficiency \cite{lindquist_comparison_1989}, low memory requirements \cite{sandia_trilinos_2017}, and high mesh quality \cite{kowalski_pde_2013}. Having such a block structure on a surface is also advantageous for tasks such as spline fitting and isogeometric analysis \cite{campen_partitioning_2017}. In the past, such meshes have been designed by hand or in an interactive environment \cite{dlr_megacads_2001}. More recently, researchers have made progress towards fully automating this process \cite{campen_partitioning_2017}. The problem of coarse quad layout generation can be summarized as determining the placement of irregular nodes of the layout, and determining the connectivity of those nodes in such a way that the resulting layout is topologically valid, coarse, and such that a mapping of the region to a quadrilateral in the plane results in low distortion.

  Various approaches have been taken to generate coarse quad partitions of surfaces, each of which address these problems of irregular node placement and determining connectivity in different ways. These methods include medial axis subdivision \cite{tam_2d_1991,gould_automated_2012}, computing the Morse-Smale complex \cite{dong_spectral_2006}, surface foliations \cite{lei_generalized_2017}, dual loops \cite{campen_dual_2012}, simplifying an existing quad mesh \cite{tarini_simple_2011,bommes_global_2011}, and more recently, cross field based approaches \cite{kowalski_pde_2013, bommes_integer-grid_2013, razafindrazaka_perfect_2015, fogg_automatic_2015, campen_quantized_2015,zhang_automatic_2016,pietroni_tracing_2016}. While in this discussion we focus on cross field based streamline tracing approaches, Campen [4] provides an excellent literature review on quad patching algorithms.

In cross field based streamline tracing approaches, irregular node placement is determined by computing singular points of the cross field, where simple singularities of positive or negative index correspond to irregular nodes of valence 3 or 5 respectively. The connectivity of the layout is then determined by either tracing out raw separatrices of the cross field, or by using the isolines of an underlying parameterization. A common problem with separatrix tracing approaches is that on a discrete geometry, singularities are never perfectly aligned. In practice, this frequently causes limit cycles and very thin regions to occur within the quad layout. This is problematic for meshing because very small mesh elements are required. Further, the \emph{base complex} of the mesh is often far more complicated than necessary, mitigating the benefits of a multi-block decomposition. Another problem is that despite the fact that in the continuum, streamlines can only cross each other orthogonally \cite{viertel_approach_2019}, numerical inaccuracies often lead to tangential crossings of streamlines. This is especially true near singularities, where large changes in direction occur over arbitrarily small lengths.

We make three primary contributions:
\begin{enumerate}
  \item We extend the diffusion generated method for cross field design in \cite{viertel_approach_2019} to curved surfaces. This method which has previously only been described in 2D tends to have good singularity placement in locations where singularities occur because of boundary curvature rather than Gaussian curvature of the surface, a common scenario for CAD surfaces.
  \item We prove that near singularities, streamlines of a cross field are hyperbolic under a conformal map. This results in a simple method to compute streamlines in the neighborhood of a singularity that prevents tangential crossings.
  \item We describe a novel algorithm for simplifying a quad partition with T-junctions such that the number of partition components strictly decreases and the number of T-junctions decreases monotonically.
\end{enumerate}
\begin{algorithm}[t]
\caption{Partitioning a surface $M$ into a coarse quad layout.} \label{alg:ps-overview}

\vspace{.2cm}

\begin{algorithmic}
\STATE{\bfseries Input:} A triangle mesh $T$ representing a surface $M$ in $\R^3$.

\vspace{.2cm}

\STATE{\bfseries Output:} A quad layout with T-junctions $Q$ partitioning $M$ into four-sided regions.

\vspace{.2cm}

\STATE 1. Compute a smooth cross field on $\mathrm{T}$.
\STATE 2. Construct a quad layout with T-junctions by tracing out separatrices of the cross field.
\STATE 3. Simplify the quad layout with T-junctions.

\end{algorithmic}
\end{algorithm}
To demonstrate these contributions, we include them in the pipeline described in algorithm \ref{alg:ps-overview}, which takes as input a triangle mesh, and outputs a coarse quad layout, possibly with T-junctions.

\subsection{Related Work}

Several researchers have taken an approach to generating quad layouts that is similar to the pipeline in algorithm \ref{alg:ps-overview}. Kowalski et al.\ \cite{kowalski_pde_2013} design a cross field by solving a PDE with a constraint applied via Lagrange multipliers. They numerically integrate streamlines from each interior singularity and boundary corner, and snap streamlines to singularities when they pass within a certain tolerance of the singularity in order to obtain a coarser quad layout. While streamline snapping works well on some examples, it is not robust in general because it can introduce tangential crossings when more than one streamline passes nearby a singularity. Fogg et al.\ \cite{fogg_automatic_2015} take a similar approach, but initialize the cross field by an advancing front method and then smooth it with the energy functional introduced by Hertzmann and Zorin \cite{hertzmann_illustrating_2000}. Rather than snapping separatrices to singularities, they allow them to pass by singularities, resulting in thin regions throughout the partition. Ray and Sokolov \cite{ray_robust_2014} and Myles et al.\ \cite{myles_robust_2014} both implement robust streamline tracing algorithms based on edge maps \cite{jadhav_consistent_2012}, and then trace out streamlines in parallel until their first crossing with another streamline, forming a motorcycle graph \cite{eppstein_motorcycle_2008}. This approach yields coarse quad patches because of the large number of T-junctions that appear in the decomposition. Because of the T-junctions, it is a non-trivial matter to assign globally consistent parametric lengths to each edge. Myles et al.\ \cite{myles_robust_2014} are able to achieve this via a heuristic method, which includes a collapse operation, similar to the one we describe in \cref{sec:partition-simplification}, used to remove edges with zero parametric length, and 2-6 cone insertion to remove zero edges which cannot be collapsed. They subsequently generate a globally consistent seamless parameterization on each quad region. This method differs from our goal in that the parameterization is not \emph{quantized} (cf. \cite{campen_quantized_2015}) and so does not correspond to a quad mesh let alone a coarse quad decomposition. Campen et al.\ \cite{campen_quantized_2015} compute quantized parameterizations on surfaces by solving a combinatorial optimization problem similar to \cite{bommes_integer-grid_2013} but leverage the structure of a motorcycle graph to determine a set of linear equality constraints which are applied to the final optimization problem, guaranteeing a valid solution which also outperforms previous parameterization methods, especially in cases with a large number of singularities or large target edge lengths. This enables them finally to extract a coarse quad mesh. Razafindrazaka et al.\ \cite{razafindrazaka_perfect_2015} trace out separatrices of a seamless parameterization to generate a graph of possible matchings between separatrices and singularities. They then formulate the problem of connecting singularities together as a minimum weight perfect matching problem. In a second paper \cite{razafindrazaka_optimal_2017} they extend their perfect matching method to work on an input quad mesh rather than a seamless parameterization. Zhang et al.\ \cite{zhang_automatic_2016} employ a similar strategy to \cite{razafindrazaka_perfect_2015}, from an input seamless parameterization they identify candidate separatrices as monotone isolines of the parameterization in areas they call \emph{safety turning areas} which are rectangular parameterizations between opposite singularities. From these candidate separatrices they choose an optimal set via a binary optimization problem. Pietroni et al.\ \cite{pietroni_tracing_2016} trace out candidate curves connecting either two singularities or a singularity to a boundary which they call \emph{field-coherent} streamlines, which can deviate from the cross field but never switch to a different direction of the field. They then select a valid subset of these candidate streamlines to form a quad layout with T-junctions by solving a binary optimization problem.

The pipeline we implement is most similar to \cite{razafindrazaka_perfect_2015,zhang_automatic_2016} and \cite{pietroni_tracing_2016} because it attempts to generate a quad layout by directly manipulating the streamlines of an underlying cross field. It also shares similarities with the methods of Tarini et al.\ \cite{tarini_simple_2011} and Bommes et al.\ \cite{bommes_global_2011}, which attempt to simplify the base complex of an unstructured quad mesh via greedy application of grid preserving operators, the difference being that our simplification method is applied directly to the separatrices of a cross field rather than an input quad mesh.

\section{Cross Field Design} \label{sec:cross-fields}

Many different methods exist to design a smooth cross field on a curved surface \cite{hertzmann_illustrating_2000,palacios_rotational_2007,ray_n-symmetry_2008,ray_geometry-aware_2009,bommes_mixed-integer_2009,crane_trivial_2010,knoppel_globally_2013,panozzo_frame_2014,jakob_instant_2015,vaxman_directional_2016,huang_extrinsically_2016,beaufort_computing_2017}, any of which might be used in algorithm \ref{alg:ps-overview}. In the interest of brevity, we will not offer a full review of the literature here but instead refer the reader to \cite{vaxman_directional_2016}. The diffusion generated method introduced in Viertel and Osting \cite{viertel_approach_2019} is appealing in the context of meshing CAD surfaces because it is comparable in speed to the fastest methods such as those described in Kn\"oppel et al.\ \cite{knoppel_globally_2013} and Jakob et al.\ \cite{jakob_instant_2015}, but also has an advantage over these methods on surfaces with boundary, which are common in CAD and meshing for FEM, because it tends to place singularities symmetrically even when the surface is flat (see \Cref{fig:disks}).

\begin{figure}[h]
  \includegraphics[width=\linewidth]{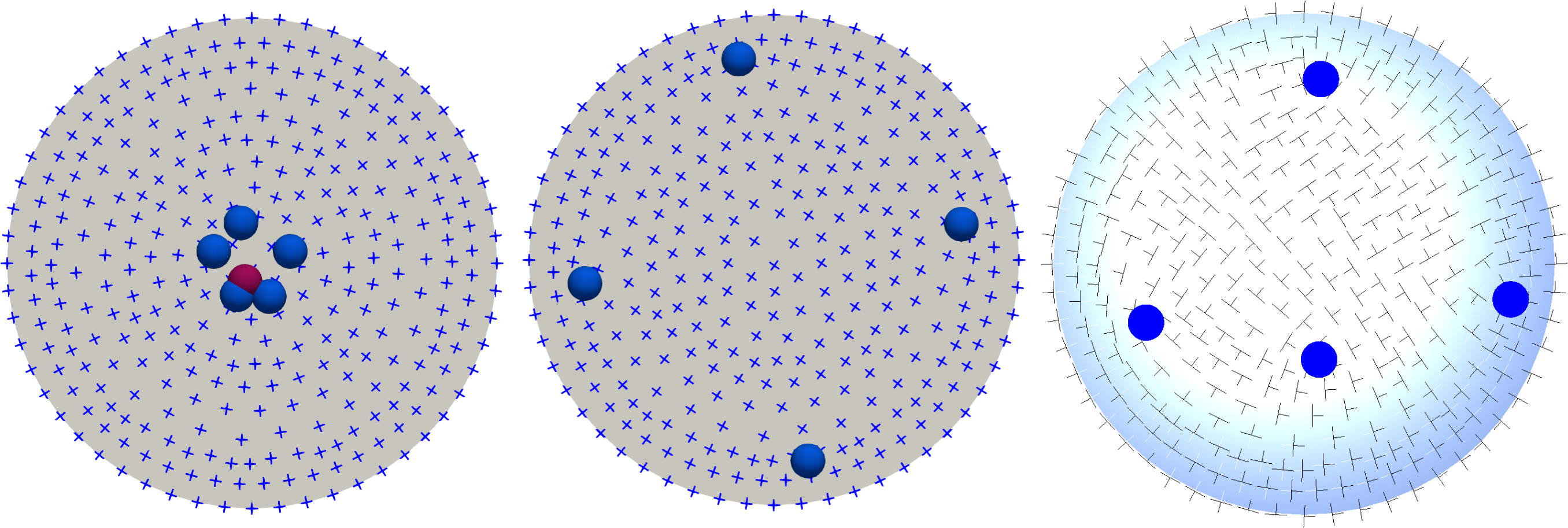}
  \caption{Left to right, a comparison of singularity placement from cross fields described in Kn\"oppel et al.\ \cite{knoppel_globally_2013}, Viertel and Osting \cite{viertel_approach_2019}, and Jakob et al.\ \cite{jakob_instant_2015}. Singularities with positive and negative index are shown in blue and red respectively.}
  \label{fig:disks}
\end{figure}

In this section, we extend the diffusion generated method for cross field design to curved surfaces. We use this method of cross field design as the first step in our pipeline (\cref{alg:ps-overview}). In \cref{sec:surfaces}, we first review some tools from the literature that have been used previously to develop cross field design methods on surfaces. In \cref{sec:discrete-MBO}, we describe the implementation details of our method.

\subsection{Cross Fields on Surfaces} \label{sec:surfaces}
The main difficulty in extending flat 2D methods to curved surfaces is the lack of a global coordinate system. In this section, we use concepts from differential geometry to formulate the cross field design problem on 2-manifolds. We consider a smooth, orientable 2-manifold $M$ embedded in $\R^3$ and endowed with the Riemannian metric induced by the the Euclidean metric on $\R^3$. Let $\mathbb{T} = \{ z \in \mathbb{C}\colon |z| = 1 \}$ be the circle group with group operation given by complex multiplication and let $\rho(4)$ be the set of the $4$th roots of unity. A \emph{cross} is an element of $C = \mathbb{T}/\rho(4)$. There is a canonical group isomorphism $R \colon C \rightarrow \mathbb{T}$ called the \emph{representation map} given by $R([c]) = c^4,$ where $c$ is any representative member of the equivalence class $[c] \in C$. We will refer to $u = c^4$ as the \emph{representation vector} for $[c]$. The \emph{inverse representation map} $R^{-1} \colon \mathbb{T} \rightarrow C$ assigns $u \in \mathbb{T}$ to $R^{-1}(u) = \left[\sqrt[4]{u}\right]$, the equivalence class of the $4$th roots of $u$.

  Let $T_p$ be the tangent space of $M$ at a point $p$. The disjoint union of all tangent spaces on $M$ is called the \emph{tangent bundle} and is denoted $TM$. For each tangent space $T_p$ we select a coordinate basis, $\{\frac{\partial }{\partial x^1}|_p,\frac{\partial }{\partial x^2}|_p\}$. We also associate with each point $p$ of $M$ a space homeomorphic to $C$, which we call the \emph{cross space} at $p$. We denote this space by $C_p$. The disjoint union of all cross spaces of $M$ defines a fiber bundle that we refer to as the \emph{cross bundle}. A section of the cross bundle, or a choice of one cross per cross space, is called a \emph{cross field} on $M$. We make the natural identification between $T_p$ and the complex plane by the map $(a,b) \mapsto a + ib$. In this way, we can identify a cross, $[c_p]$, in $C_p$ as an unordered set of four orthogonal unit vectors in $T_p$, which we call the cross \emph{component vectors}. This also allows us to define a representation map $R_p$ at each point with respect to the local coordinate basis, and a representation vector $u_p = R_p([c_p])$, which we identify with the corresponding unit tangent vector in $T_p$. For simplicity, we will use complex notation for equations throughout the paper.

In 2D cross field design, the goal of designing a smooth cross field is often formulated as designing a harmonic representation vector field $u$ \cite{viertel_approach_2019,vaxman_directional_2016}. That is, to minimize the Dirichlet energy
\begin{equation} \label{eq:energy}
  E[u] \coloneqq \frac{1}{2}\int_M{|\nabla{u}|^2dA}
\end{equation}

\noindent with the constraint that $|u| = 1$ at each point of the domain. We note that in general, this problem is ill-posed; however, generalized solutions exist if a finite number of singular points are removed from $M$ \cite{viertel_approach_2019,bethuel_ginzburg-landau_1994}.

This strategy can be extended to surfaces by replacing the gradient operator $\nabla$ with the appropriate \emph{connection} on the tangent bundle. The Levi-Civita connection provides a way to compare vectors on the tangent bundle that preserves the notion of inner product between tangent spaces. That is, if $P_{pq}$ is the parallel transport function for the Levi-Civita connection between $T_p$ and $T_q$, and if ${v}_1, {v}_2 \in T_p$, then
\[ \langle{v}_1,{v}_2 \rangle_{T_P} = \langle P_{pq}({v}_1),P_{pq}({v}_2) \rangle_{T_q}.\]

Visually, the effect of using this connection in equation \ref{eq:energy} is that a minimizing vector field appears smooth (see \Cref{fig:LC-smooth}).
\begin{figure*}[h]
  \includegraphics[width=\textwidth]{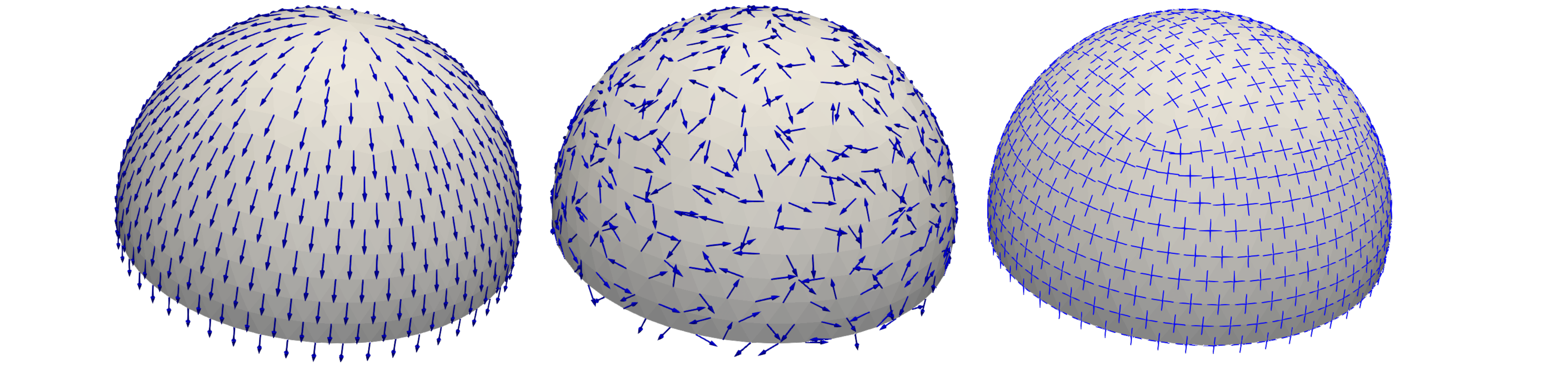}
  \caption{A comparison of fields and connections. {\bf Left:} A smooth vector field as measured by the Levi-Civita connection. {\bf Center:} A smooth vector field as measured by the connection $\nabla^Q$. {\bf Right:} The cross field corresponding to the vector field in the center.}
  \label{fig:LC-smooth}
\end{figure*}

In order to extend the strategy from 2D, using equation \ref{eq:energy} to design a smooth cross field, we seek a connection, $\nabla^Q$, on the tangent bundle with corresponding parallel transport function, $Q$, appropriate for transporting \emph{representation vectors}. We choose $Q$ in such a way that the component vectors of the corresponding crosses are transported by $P$, the parallel transport function corresponding to the Levi-Civita connection. Let $\gamma_{pq}(t) \colon [0,1] \rightarrow M$ be a Levi-Civita geodesic connecting points $p$ and $q$ such that $\gamma_{pq}(t_p) = p$, $\gamma_{pq}(t_q) = q$, and $t_p, t_q \in (0,1)$. Let $\phi_p$ be the signed angle from the velocity vector $\gamma'_{pq}(t_p)$ to $\frac{\partial }{\partial x^1}|_p$ and let $\phi_q$ be the signed angle from $\gamma'_{pq}(t_q)$ to $\frac{\partial }{\partial x^1}|_q$. Then if $\phi_{pq} = \phi_p - \phi_q$, $P_{pq}({v}) = e^{i\phi_{pq}}{v}$ gives the Levi-Civita parallel transport function, $P_{pq}$, in coordinates with respect to the basis $\{\frac{\partial }{\partial x^1}|_p,\frac{\partial }{\partial x^2}|_p\}$.

 Then, $Q_{pq}$, the parallel transport function for \emph{representation vectors} between $p$ and $q$, must satisfy
\begin{equation} \label{eq:parallel-transport}
  Q_{pq}(R_p([c_p])) = R_q([P_{pq}(c_p)]).
\end{equation}

We can write $c_p = e^{i(\theta_p +2k\pi/4)}$ for some $k \in \{0,1,2,3\}$ where $\theta_p$ is the signed angle from $\frac{\partial }{\partial x^1}|_p$ to one of the component vectors of $[c_p]$. We can now write equation \ref{eq:parallel-transport} as
\[ Q_{pq}(e^{4i\theta_p}) = e^{4i(\theta_p + \phi_{pq})}.\]

\noindent It follows that $Q_{pq}({v}) = e^{4i\phi_{pq}}{v}$. In addition, we can define parallel transport on the cross bundle from $C_p$ to $C_q$ by $R^{-1}_q \circ Q_{pq} \circ R_p$. We can now use $\nabla^Q$ in equation \ref{eq:energy} to define a smooth cross field. In the following section, we define the discrete Laplace equation corresponding to this energy and describe the diffusion generated method for cross field design in detail.

\subsection{Discrete Formulation} \label{sec:discrete-MBO}

On each node $n_i$ of the input triangle mesh, $\mathrm{T}$, a normal vector, $\vec{n}_i$, is computed as an average of the vectors normal to each adjacent face weighted by the tip angle at the node. This normal vector in turn defines the tangent space, $T_i$, at node $n_i$. We then arbitrarily select a vector in $T_i$, which we assign to be $\frac{\partial }{\partial x^1}|_i$. Then $\frac{\partial }{\partial x^2}|_i$ is the vector such that $\frac{\partial }{\partial x^1}|_i \times \frac{\partial }{\partial x^2}|_i = \vec{n}_i$.

Let $e_{ij}$ be the edge connecting nodes $n_i$ and $n_j$. We compute the value $\phi_i$ as the signed angle between the projection of $e_{ij}$ onto $T_{i}$ and $\frac{\partial }{\partial x^1}|_{i}$. We then store the value $\phi_{ij} = \phi_j - \phi_i$, on the edge $e_{ij}$.

We now define the discrete parallel transport function $Q_{ij} \colon T_i \rightarrow T_j$ by $Q_{ij}(u) = e^{4i\phi_{ij}}u$, which parallel transports representation vectors across $e_{ij}$. This allows us to compare two representation vectors $u \in T_i$ and $v \in T_j$ by
\[|v - Q_{ij}(u)|^2. \]

In order to state a well-defined problem, we apply a Dirichlet condition. In the case of closed manifolds, we arbitrarily fix the orientation of a single cross. In the case of a bounded manifold, we apply a Dirichlet boundary condition by fixing $u_i$ on each boundary node. We make the convention that the cross field index of a boundary node (see \cite{viertel_approach_2019}) is assigned according to \Cref{tab:boundary-index}.
\begin{table}[t]
  \caption{Boundary index assignment} \label{tab:boundary-index}
  \begin{tabular}{p{2cm}p{2.4cm}p{2cm}p{4.9cm}}
    \hline\noalign{\smallskip}
    Interior Angle & Index \\
    \noalign{\smallskip}\hline\noalign{\smallskip}
    $(0, \frac{3\pi}{4})$ & $\frac{1}{4}$ \\
    $[\frac{3\pi}{4}, \frac{5\pi}{4}]$ & 0 \\
    $(\frac{5\pi}{4}, \frac{7\pi}{4}]$ & $-\frac{1}{4}$ \\
    $(\frac{7\pi}{4}, 2\pi)$ & $-\frac{1}{2}$ \\
    \noalign{\smallskip}\hline\noalign{\smallskip}
  \end{tabular}
\end{table}

For a boundary node $n_i$ of index 0 or $-\frac{1}{2}$, we compute the outward pointing unit normal vector of each boundary edge adjacent to $n_i$ lying in the plane of the adjacent boundary face. We then project these vectors into the tangent plane at $n_i$, and bisect them with a unit vector, $d_i$, which averages the directions of the facet normals. Since we want to align the cross to this vector, we set $u_i = d_i^4$. For nodes of index $\pm$ $\frac{1}{4}$, $d_i$ is first rotated $\pi/4$ radians about $\mathbf{n}_i$ in the positive direction before computing $u_i$.

\subsubsection{The Diffusion Generated Method for Cross Field Design}

Intuitively, we would like to solve until stationarity the time dependent problem
\begin{align} \label{eq:time-dependent}
  u_t(t,x) &= \Delta u(t,x) && x \in M\nonumber\\
  u(t,x) &= g(x) && x \in \partial M\\
  u(0,x) &= u^0(x) && x \in M\nonumber
\end{align}

\noindent with the constraint that $|u(x)| = 1$ pointwise, rather than directly solving the stationary problem with the same constraint. This key difference allows one to avoid the necessity of using a non-linear solver because the pointwise constraint can be enforced simply by normalizing the solution in between discrete time steps (see \cite{ruuth_diffusion-generated_2001,laux_analysis_2019}). We proceed by defining a discrete Laplacian operator, $\Delta_Q$ corresponding to the connection $\nabla^Q$.
\begin{equation} \label{eq:discrete-laplacian}
  \Delta_Q(u)|_i = \frac{1}{|\mathcal{N}(n_i)|} \sum_{n_j \in \mathcal{N}(n_i)}{\left(u_j - Q_{ij}(u_i)\right)}
\end{equation}

\noindent where $\mathcal{N}(n_i)$ is the one-ring neighborhood of $n_i$ and $|\mathcal{N}(n_i)|$ is the area of that neighborhood. This discrete Laplacian operator allows us to assemble a discrete diffusion equation using a backward Euler time discretization;
\begin{equation}
  (I - \tau A)u = u_0 + \tau b
\end{equation}

\noindent where $A$ is the matrix form of $\Delta_Q$. We iteratively solve this equation for a small time step $\tau$ and pointwise renormalize the solution between each iteration. The algorithm is described in detail in algorithm \ref{alg:MBO}.

\begin{algorithm}[t]
\caption{A diffusion generated method for designing smooth cross fields} \label{alg:MBO}
\vspace{.2cm}

  \begin{algorithmic}
    \STATE Let $u^0$ be the solution to $Au = b$.
    \STATE Fix $\tau$, $\delta$, and set $k = 0$.
    \WHILE {$\| u^k - u^{k-1} \| > \delta $,}
    \STATE Solve the discrete diffusion equation,
    \begin{equation}
      (I - \tau A)u^{k+1} = u^k + \tau b
    \end{equation}
    \FOR {$j \in [0, n]$}
    \STATE Set $u^{k+1}_j = \frac{u^{k+1}_j}{|u^{k+1}_j|}$\\
    \ENDFOR
    \STATE $k++$
    \ENDWHILE
    \vspace{.2cm}
  \end{algorithmic}
\end{algorithm}

\section{Generation of a Quad Layout with T-junctions} \label{sec:quad-layout}

After designing a cross field on a triangle mesh, the next step in \cref{alg:ps-overview} is to construct a quad layout with T-junctions. This is accomplished in two steps:
\begin{enumerate}
  \item Determine singularity locations and ports.
  \item Trace out separatrices of the cross field.
\end{enumerate}

In \cref{sec:singularities-and-ports}, we detail the first step of this process. In \cref{sec:separatrices}, we describe our approach to streamline tracing, including our novel method for computing the trajectory of a streamline in the neighborhood of a singularity. In \cref{sec:tangential-crossings}, we discuss stopping conditions for the streamline tracing algorithm as well as a heuristic method for handling tangential streamline crossings that occur commonly when tracing streamlines via numerical integration.

\subsection{Singularity and Port Detection} \label{sec:singularities-and-ports}

Singularity and port detection in a cross field defined per node is well documented in the literature \cite{kowalski_pde_2013,fogg_automatic_2015,viertel_approach_2019,knoppel_globally_2013,jakob_instant_2015}. We use methods that have been developed previously, but include a description here for completeness.

\subsubsection{Matchings}
  Across each edge, we assume that crosses make the smallest rotation possible. This is called the \emph{principle matching} of the crosses. If $u_i = e^{4i\theta_i}$, then the change in cross orientation between two nodes $n_i$ and $n_j$, denoted $\Delta_{ij}$, is the number between $-\frac{\pi}{4}$ and $\frac{\pi}{4}$ given by
\[  \Delta_{ij} = (\theta_j - (\phi_{ij} + \theta_i)) (\bmod\:\pi/2) - \pi/4. \]

\subsubsection{Singularity Detection}

The \emph{index} of a triangle, $t_{ijk}$, with nodes $n_i$, $n_j$, and $n_k$ is the number given by
\[ \Index{(t_{ijk})} = \frac{\Delta_{ij} + \Delta_{jk} + \Delta_{ki} - \phi_{ij} - \phi_{jk} - \phi_{ki}}{2\pi}. \]

Practically, this is the number of turns that a cross makes while circulating the triangle, which is always an integer multiple of $\frac{1}{4}$. A triangle is \emph{singular} if its index is non-zero. Summing the changes in cross orientation along each edge, we compute the total circulation of the cross around the triangle. On a flat surface, the total circulation must be a multiple of $\pi/2$; however, on curved surfaces, this is not the case. In general, a vector in the tangent bundle of $M$ parallel transported along a closed curve does not always return to the same orientation after circulating the curve. To mitigate the effects of \emph{holonomy} while transporting a cross around the triangle, we subtract from the total circulation the angle defect that occurs when parallel transporting a vector around the triangle via our discrete connection. This angle defect is given by $\phi_{ij} + \phi_{jk} + \phi_{ki}$. Subtracting the angle defect from the total circulation leaves us with a number that is a multiple of $\pi/2$, from which we compute the index.

After we determine singular triangles, we approximate the location of the singularity within the triangle by taking the barycenter of the triangle. We choose an arbitrary node, $n$, on the triangle and rotate the cross component vectors at that node into the plane of the face. We use the ray starting at the barycenter and passing through $n$ as a reference axis, and compute the angle $\alpha$ that the nearest cross component vector makes with the reference axis. We then compute the angles where streamlines exit the singularity (ports) as $\alpha + \frac{2\pi k}{4-d}$ where $d/4$ is the index of the singularity (see \cite{viertel_approach_2019}).

\subsection{Separatrix Tracing} \label{sec:separatrices}

After designing a cross field on $T$, we build an initial quad partition by tracing separatrices of the cross field. Below, we describe the methods used for tracing separatrices in non-singular and singular triangles.

\subsubsection{Non-Singular Triangles} \label{sec:non-singular-triangles}

In non-singular triangles $t_{ijk}$, we project the crosses on each corner of $t_{ijk}$ into the plane of the triangle. We define a reference coordinate axis in that plane, and compute a representation vector for each cross with respect to this reference coordinate axis. We interpolate the argument of the representation vector linearly over the triangle. This is possible because in a non-singular triangle, the argument is continuous as it circulates the boundary. Using this interpolation, we trace out the streamlines using Heun's method \cite{kowalski_pde_2013}.

\subsubsection{Singular Triangles} \label{sec:singular-triangles}

Previously, no streamline tracing methods have been suitable for accurately tracing streamlines in the neighborhood of a singular triangle. Numerical integration methods such as Heun's method and other Runge-Kutta methods \cite{kowalski_pde_2013,fogg_automatic_2015,viertel_approach_2019,zhang_vector_2006} are inaccurate in the neighborhood of a singularity since cross directions can change arbitrarily fast. They frequently compute discretizations of the streamline that ``cut the corner'' rather than traversing around the singularity as they should (see \Cref{fig:cut-the-corner}). Streamline tracing methods based on edge maps \cite{ray_robust_2014,myles_robust_2014} are guaranteed to never result in tangential crossings, but since paths are discretized by straight lines through each triangle, they are limited in their ability to resolve the path of a trajectory around a streamline by the number of triangles that meet at a singular point.
\begin{figure}[h]
  \begin{center}
    \includegraphics[width=0.15\textwidth]{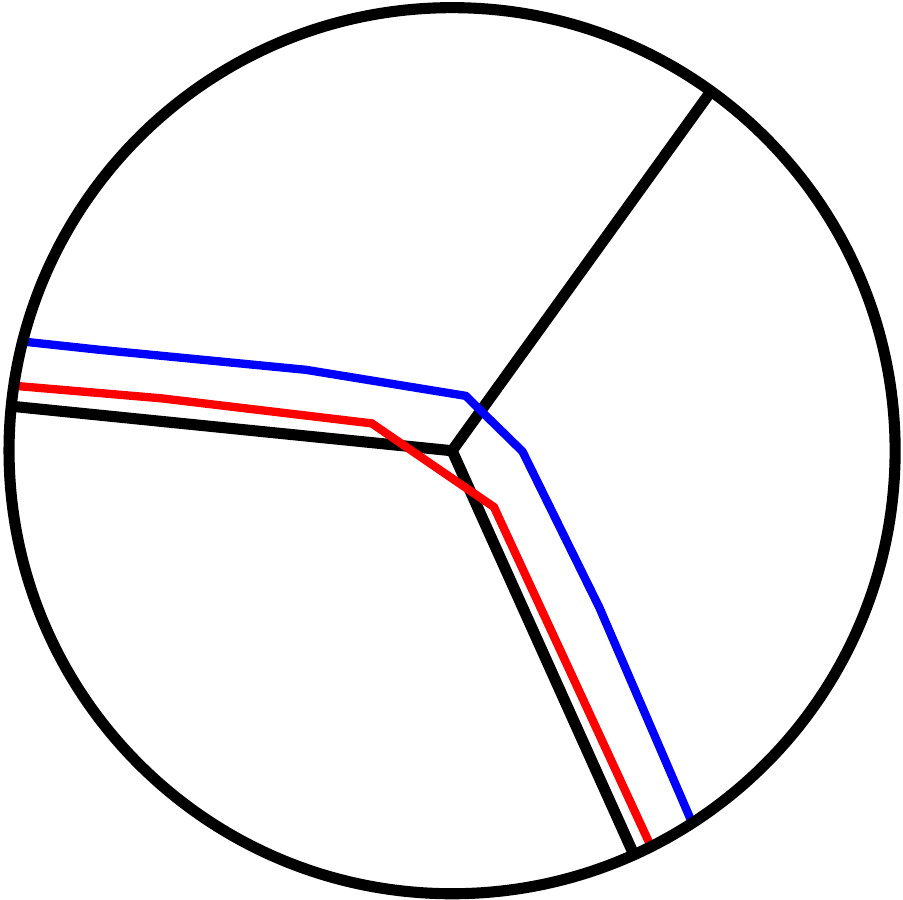}
  \end{center}
  \caption{An illustration of a numerically traced streamline that ``cuts the corner'' and does not traverse around the singularity.}
  \label{fig:cut-the-corner}
\end{figure}

Here we develop a new method to accurately trace streamlines within the neighborhood of a singularity to any predefined resolution. This method guarantees that no tangential intersections of streamlines will occur within the neighborhood. We first prove that in $\mathbb{R}^2$, trajectories of streamlines through the neighborhood of a singularity are hyperbolas under a conformal transformation.

Let $f$ be a canonical harmonic cross field (see \cite{viertel_approach_2019}) on a domain $D \subset \R^2$. Let $a \in \R^2$ be the location of a singularity of index $\frac{d}{4}$ where $d$ is an integer $\leq 1$. Consider the open ball $B(a,r_0)$ of radius $r_0 > 0$ centered at $a$. We seek an approximation for the trajectory of an arbitrary streamline passing through a point $q \neq a$ in $B(a,r_0)$.

The cross field, $f(z)$, partitions $B(a,r_0)$ into $4 - d$ evenly angled sectors bounded by separatrices of the cross field \cite{viertel_approach_2019}. In each sector, the cross field defines a local $(u,v)$ parameterization. Let $S$ be the open sector containing $q$, and let $s_0$ and $s_1$ be the separatrices bounding $S$ ordered counterclockwise. Because the cross field defines a local $(u,v)$ parameterization on $S$, there are two streamlines passing through $q$, one crossing $s_0$ orthogonally, the other crossing $s_1$ orthogonally. Without loss of generality, we consider $\gamma$, the streamline crossing $s_1$ orthogonally. This streamline crosses $s_1$ into $S'$, the open sector adjacent to $S$ that is also bounded by $s_1$, and then exits $B(a,r_0)$ (see \Cref{fig:hyperbolic}).
\begin{figure}[h]
  \includegraphics[width=\linewidth]{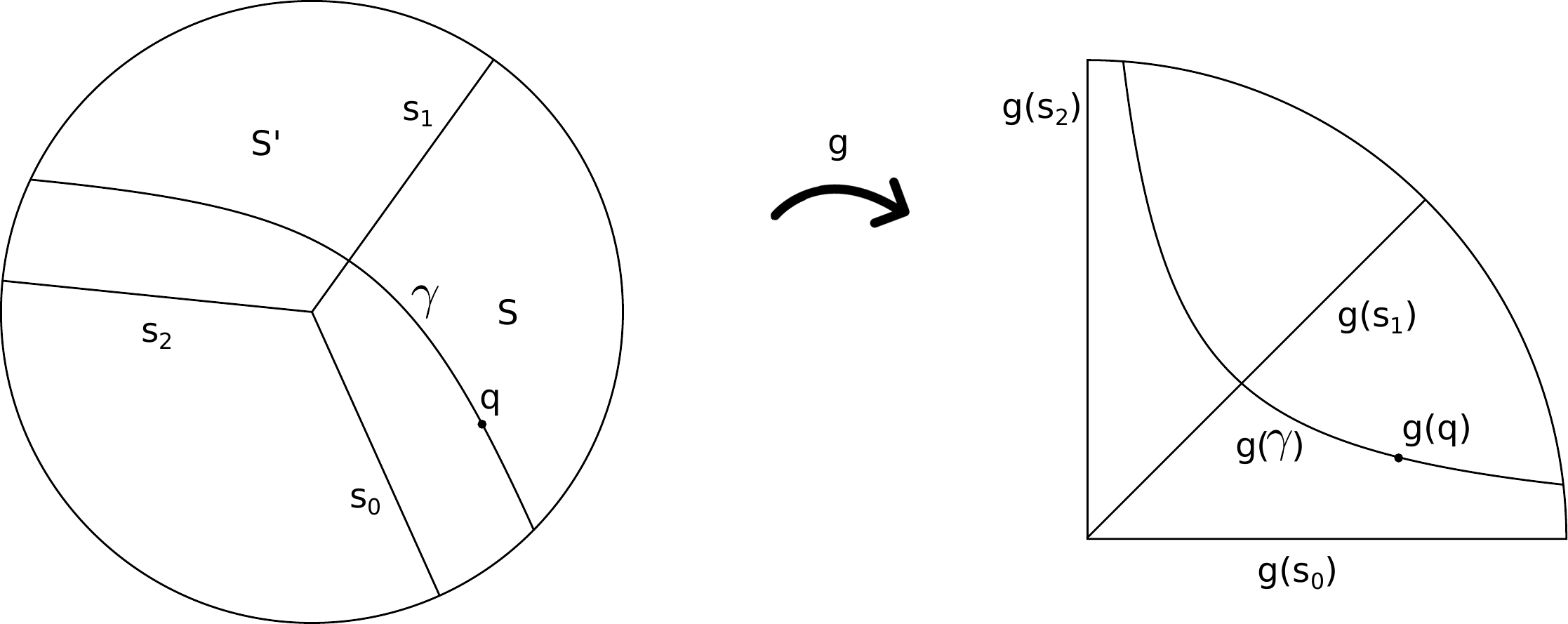}
  \caption{Streamlines in a neighborhood of a singularity become hyperbolas under a conformal map $g$.}
  \label{fig:hyperbolic}
\end{figure}

By \cite{viertel_approach_2019}, the cross field in $B(a,r_0)$ can be written as \[f(z) = e^{i(\frac{d\theta}{4} + \frac{2k\pi}{4})} + o(r)\] where $d/4$ is the index of the singularity, $z=re^{i\theta}$, $k \in \{0,1,2,3\}$, and $\theta = 0$ corresponds to $s_0$. For $r < r_0$ where $r_0$ is sufficiently small, we make the approximation \[f(z) = e^{i(\frac{d\theta}{4} + \frac{2k\pi}{4})}.\] Streamlines of the cross field are given by \[ z' = f(z). \] Since we are looking for the streamline crossing through $s_1$, we consider $k = 0$. Thus, we are looking for the set $C = \{z(t) \in B(a,r_0)\, |\, t \in (t_a,t_b)\}$ where $z(t)$ on $(t_a,t_b)$ is the solution to the problem
\begin{align}\label{eq:ODE}
  z' &= e^{i\frac{d\theta}{4}}\\
  z(0) &= (r_q,\theta_q)
\end{align}

\noindent in $D = \{z = re^{i\theta}\, |\, r \in (0,r_0),\, \theta \in (0,\frac{4\pi}{4-d})\}$.

\begin{proposition} \label{prop:hyperbolic-trajectory}
  $C = \{(x + iy)^{-(4-d)/8}\, |\, xy = A,\, x \in I_x\}$ for some constant $A$ on some interval $I_x$.
\end{proposition}

\begin{proof}
  % \smartqed
  Consider a differentiable curve in $D$ given by $z(t)$ for $t \in (a,b)$. Consider the function $g(z) = z^{(4-d)/8}$ that maps $D$ to $\tilde{D}$, the sector of the upper right quadrant given by $\{w = \rho e^{i\varphi}\, |\, \rho \in (0,\rho_0 = r_0^{(4-d)/8}),\, \varphi \in (0,\frac{\pi}{2})\}$ (see \Cref{fig:hyperbolic}). Let $w(t) = g(z(t))$ for $t \in (a,b)$. Taking the derivative of both sides, we have
  \[w'(t) = g'(z(t))z'(t)\]

  \noindent since $g'(z) \neq 0$ in $D$, we have
  \begin{align*} \label{eq:conformal-relation}
    \arg{(g'(z(t))z'(t))} &= \arg{g'(z(t))} + \arg{z'(t)}\\
                          &= \left(\frac{4-d}{8} - 1\right)\theta + \arg(z'(t))
  \end{align*}

  In the case that $z(t)$ is a solution of equation \ref{eq:ODE}, we have
  \begin{align*}
    \arg{w'(t)} &= \frac{d\theta}{4} + \left(\frac{4-d}{8} - 1\right)\theta\\ 
                &= -\frac{(4-d)\theta}{8} = -\varphi
  \end{align*}

  Thus $w'(t) = \alpha(t)e^{-i\varphi}$ for some function $\alpha(t)$. Writing $w(t) = x(t) + iy(t)$, we have $x'(t) = \alpha(t)\cos(\varphi)$, $y'(t) = -\alpha(t)\sin(\varphi)$. Thus
  \[\frac{dy}{dx} = -\tan(\varphi) = -\frac{y}{x} \implies y = \frac{A}{x}\]

\noindent for some constant $A$. This equation describes the family of hyperbolas in the first quadrant with asymptotes at $\varphi = 0$ and $\varphi = \frac{\pi}{2}$. The curve from this family passing through the point $g(q) = \rho_q e^{i\varphi_q}$ is given by $\{x + iy\, |\, xy = A_q\}$, where $A_q = \rho_q^2\sin{\varphi_q}\cos{\varphi_q}$. The curve $C$ can be recovered by taking the inverse image of this set under the mapping $g$, that is $C = \{(x + iy)^{-(4-d)/8}\, |\, xy = A_q,\, x \in I_x\}$ where $I_x$ is an interval such that $(x+iy)^{-(4-d)/8} \in B(a,r_0)$ for $x \in I_x$.
% \qed
\end{proof}

\Cref{prop:hyperbolic-trajectory} provides a simple method for computing the trajectory of a streamline through a singular triangle. We make the assumption that within the triangle, the estimate
\begin{equation*}
  f(z) \approx e^{i\left(\frac{d\theta}{4} + \frac{2k\pi}{4}\right)}
\end{equation*}

\noindent holds. Here again $\theta = 0$ corresponds to the nearest separatrix clockwise from $q$. Making this assumption, we simply compute points along the hyperbola $xy = A_q$, and take the inverse image of each point. We use these points as discretization points of the streamline so long as they lay within the singular triangle. Since hyperbolas are convex, and $g^{-1}$ preserves the order of points along rays, in order to guarantee that two streamlines don't intersect tangentially, it is sufficient to evaluate the points of the hyperbola along predefined rays from the singular point.

\subsection{Partition Construction and Tangential Crossings} \label{sec:tangential-crossings}
By the Poincar\'e-Bendixson theorem for manifolds \cite{schwartz_generalization_1963}, streamlines of a cross field on a bounded manifold $M$ can do one of the following:
\begin{enumerate}
    \item Connect one or more singularities in a homoclinic or heteroclinic orbit
    \item Exit the boundary
    \item Approach a limit cycle
    \item Approach a limit set that is all of $M$. In this case, $M$ must be a torus.
\end{enumerate}

Because we are tracing out separatrices on a discrete mesh, in practice, they never line up perfectly with another singularity, so homoclinic and heteroclinic orbits will never occur. Thus streamlines in the discrete case can only either exit the boundary, or continue forever approaching either a limit cycle, or a limit set that is all of $M$. In order to generate a quad partition via separatrix tracing, separatrices that continue forever must be cut off after crossing some separatrix orthogonally. In practice, we use two stopping conditions: separatrices are traced until they either exit the boundary or cross the same separatrix more than once. The second condition is a simple way to eliminate the possibility of tracing out a separatrix forever, but can potentially create T-junctions on separatrices that would eventually exit the boundary.

When tracing streamlines using a numerical method such as Heun's method, there is no guarantee that streamlines won't cross each other or exit the boundary tangentially. This becomes especially problematic along boundaries of meshes where the underlying geometry has high curvature but few triangles, resulting in few crosses that are actually aligned with the discrete boundary of the triangle mesh. Tangential crossings are problematic because the regions produced via separatrix tracing are no longer guaranteed to be four-sided. Assuming a sufficiently fine triangle mesh along the boundary such that no separatrices exit tangentially, we observe in practice that tangential crossings on the interior typically occur in one of two cases. The first case is when one or more separatrices that approach a limit cycle are traced out for several rotations around the limit cycle. This problem is virtually eliminated by our approach of cutting off separatrices after they cross the same separatrix more than once.

The second case where tangential crossings occur is when there is a very small misalignment of singularities, such that two different separatrices follow virtually the same path. If the two separatrices are heading in opposite directions when the crossing occurs, then this problem can easily be fixed by cutting both separatrices off at the tangential crossing and combining them into a single separatrix, now connecting the two singularities in a heteroclinic orbit. If both separatrices are traveling in the same direction when they cross tangentially, there is no analogous simple operation to combine the two. However, we have observed that in practice, this almost always occurs when one of the separatrices passes very near the singularity where the other began. To mitigate the occurrence of tangential crossings when both separatrices are traveling in the same direction, we add a third stopping criteria for tracing separatrices: we cut off any separatrix at a T-junction inside a singular triangle when it orthogonally crosses a separatrix leaving leaving the singularity. This third stopping condition greatly reduces the number of tangential crossings that occur when both streamlines are traveling in the same direction.

\section{Partition Simplification} \label{sec:partition-simplification}

The misalignment of singularities when tracing out separatrices often results in small regions and limit cycles in the initial partition that would not exist if the separatrices coincided. In this section, we present a robust algorithm to simplify the partition obtained by naive separatrix tracing. The central step in the algorithm is an operation that extends the chord collapse operation for quad meshes \cite{borden_hexahedral_2002,daniels_quadrilateral_2008} to quad layouts with T-junctions. A similar collapse operation appears in Myles et al.\ \cite{myles_robust_2014}.

A \emph{chord} in a quad mesh is a maximal sequence of quads, $q_1, q_2, \dots, q_n$ such that $q_i$ is adjacent to $q_{i+1}$, and $q_{i-1}$ and $q_{i+1}$ are on opposite sides of $q_i$. \Cref{fig:chord} shows a chord of a quad mesh highlighted in blue. A partition obtained from streamline tracing is a quad layout with T-junctions, or a \emph{T-layout} for short. We say that each component of a T-layout has four total \emph{sides}, 2 pairs that are opposite each other. A side consists of at least one edge or more when T-junctions occur on that side. A \emph{chord of a T-layout} is a maximal sequence of components, $c_1, c_2, \dots, c_n$ such that $c_i$ is adjacent to $c_{i+1}$, $c_{i-1}$ and $c_{i+1}$ are on opposite sides of $c_i$, and no T-junction exists between $c_i$ and $c_{i+1}$. \Cref{fig:t-chords} top shows various chords in a T-layout.
\begin{figure}[h!]
  \centering
  \includegraphics[width=0.35\linewidth]{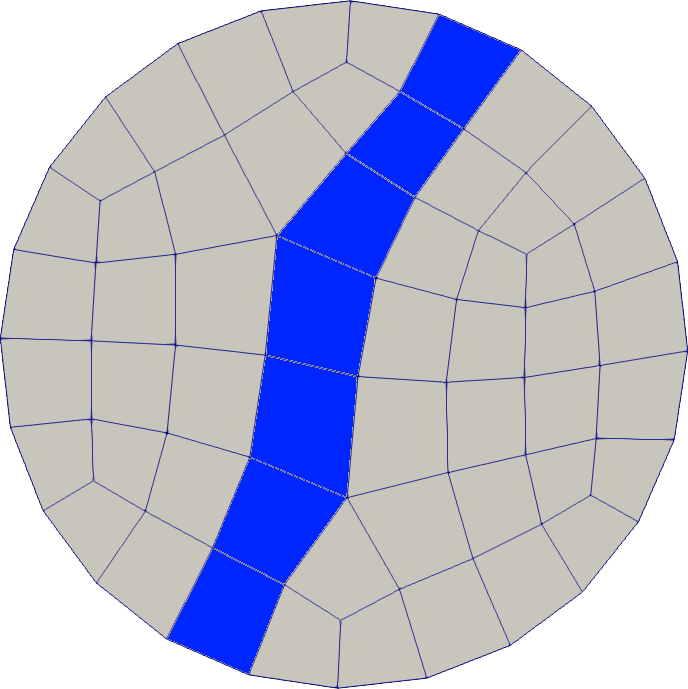}
  \caption{A chord in a quad mesh.}
  \label{fig:chord}
\end{figure}

\begin{figure}[h]
  \begin{center}
    \includegraphics[width=0.8\linewidth]{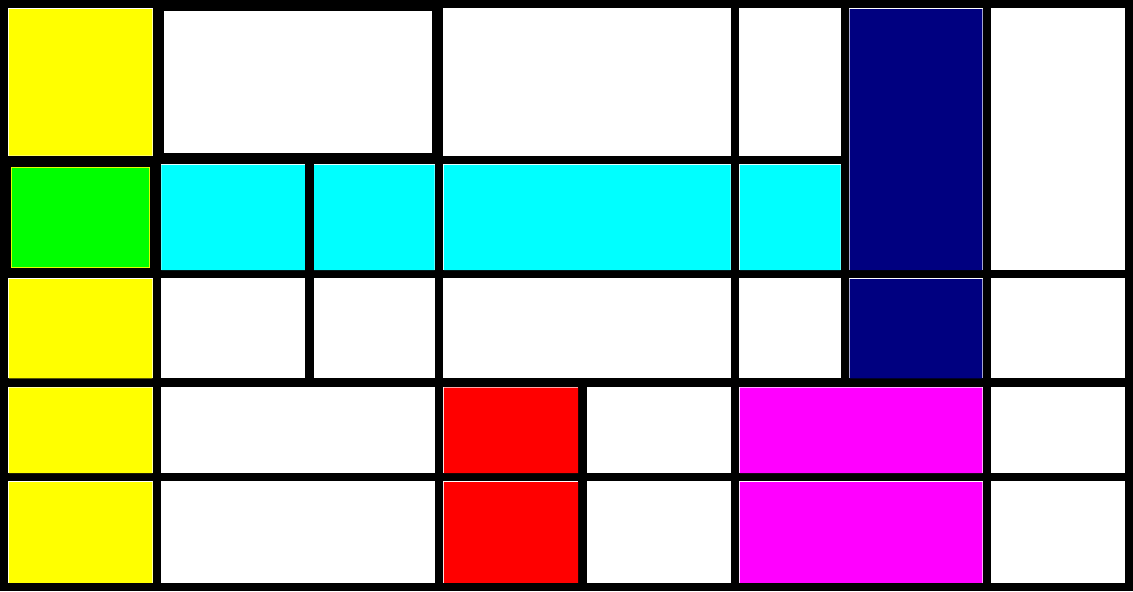}\\
    \includegraphics[width=0.8\linewidth]{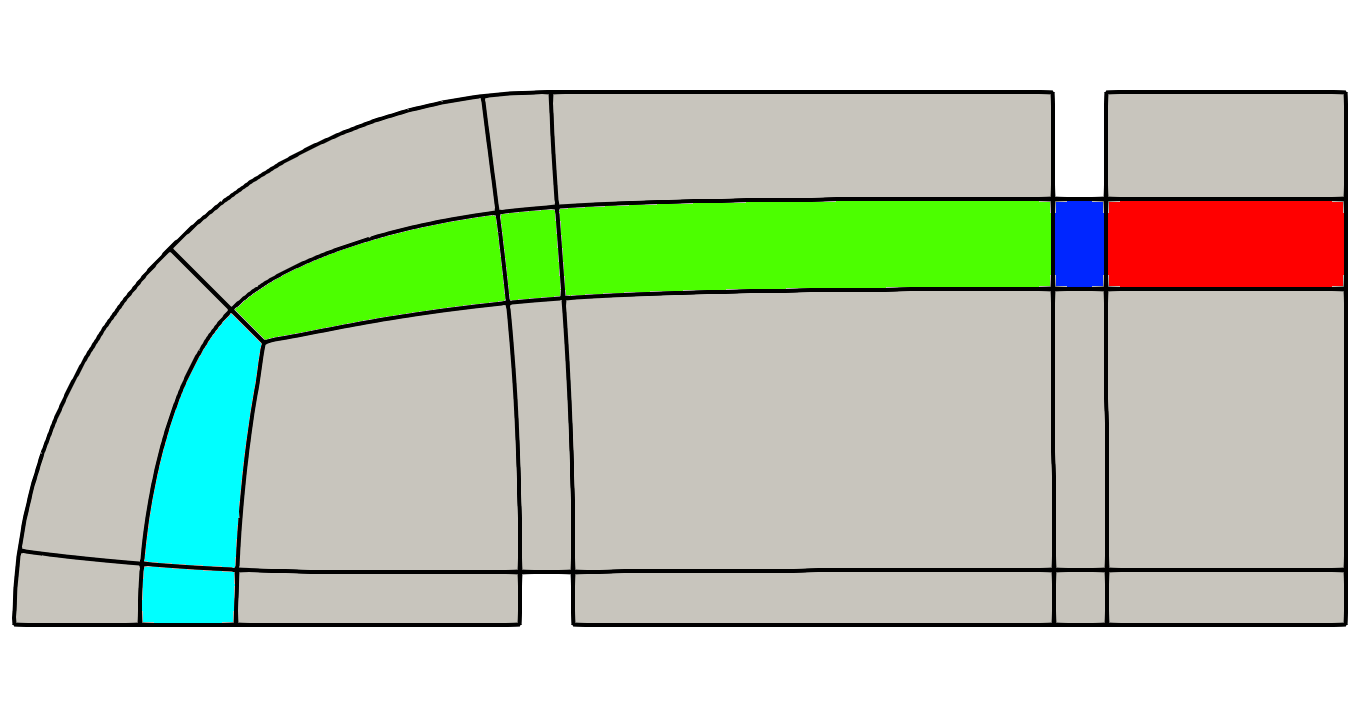}
  \end{center}
  \caption{Illustrations of chords in two T-layouts. {\bf Top:} Chords in a T-layout shown in various colors. The yellow and cyan chords overlap on the green component, illustrating how each component is part of two chords. {\bf Bottom:} The four patches of a chord highlighted in cyan, green, blue, and red.}
  \label{fig:t-chords}
\end{figure}

We call the set of edges shared by two partition components in a chord the \emph{transverse rungs} of a chord. In the case that a chord begins or ends at a T-junction or on a boundary, we also include the first and last set of edges as transverse rungs of the chord. We also say that a chord has two \emph{longitudinal sides} that consist of all the edges of the partition components that are orthogonal to the transverse rungs.

A \emph{patch} is a maximal subset of consecutive components of a chord such that singularities only occur on the first and last transverse rungs. A chord is partitioned into one or more patches, and singularities can occur only on the corners of patches. \Cref{fig:t-chords} bottom shows the patches of a chord.

    \begin{figure}[h!]
      \begin{center}
        \includegraphics[width=\linewidth]{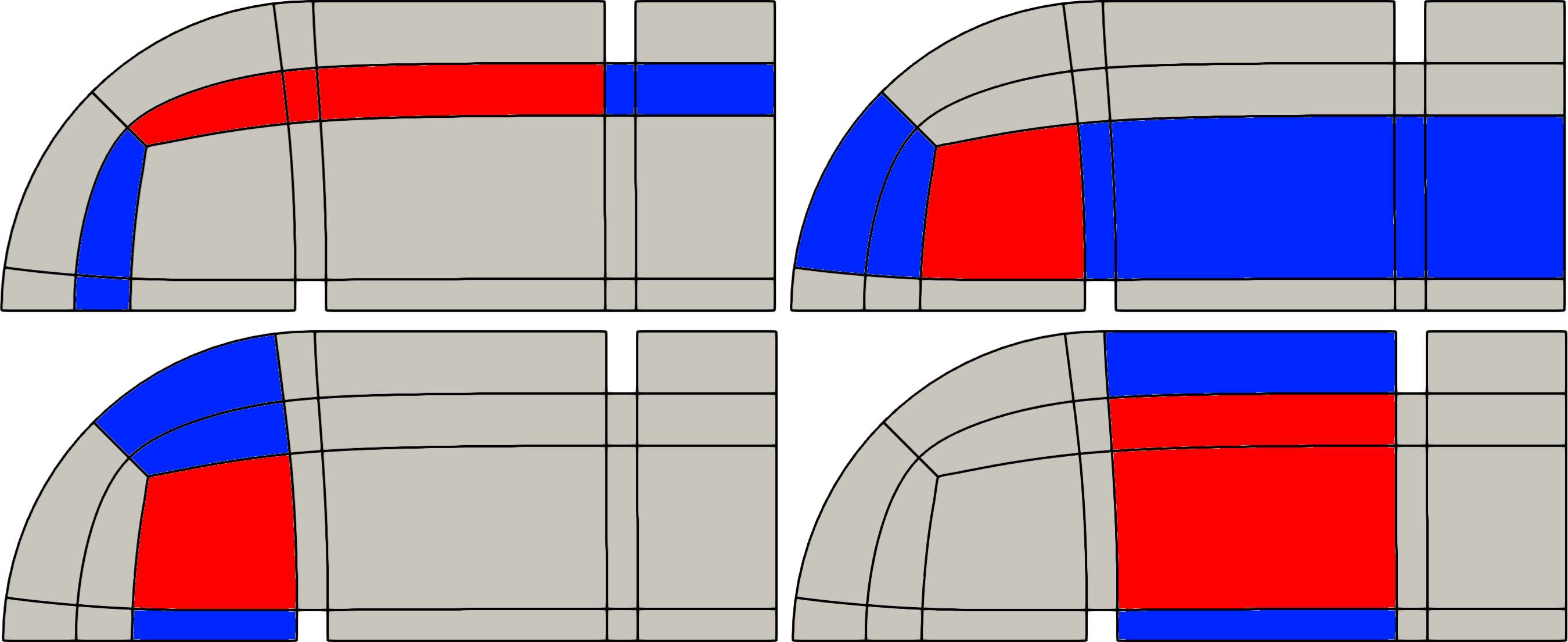}
      \end{center}
        \caption{The four collapsible chords of a partition. Zip patches are highlighted in red and non-zip patches are highlighted in blue.}
        \label{fig:collapsible}
    \end{figure}
Our definition of chord collapse on a T-layout is motivated by the goal of simplifying the partition by removing one separatrix from each of two singularities, and then connecting singularities together by a single curve. We say that a patch of a chord in a T-layout is \emph{collapsible} if it satisfies the following:
\begin{enumerate} \label{enum:collapsible}
  \item No singularities are connected across any transverse rung of the patch.
  \item No singularity is connected to a boundary across any transverse rung of the patch.
  \item If the patch starts or ends at a T-junction, then one of the following must be satisfied:
  \begin{enumerate}
    \item The node opposite the T-junction on the same transverse rung is a singularity.
    \item The node opposite the T-junction on the same transverse rung is another T-junction with the same orientation.
    \item The node on the opposite corner of the patch from the T-junction is a singularity.
  \end{enumerate}
\end{enumerate}

\noindent We say that a chord is \emph{collapsible} if all of its patches are collapsible.

The first and second conditions prevent the possibility of having to combine two singularities into a single one or move a singularity to the boundary. They reflect an assumption that throughout the simplification process, we would like to keep the singularity set of the cross field and only modify the connectivity of the singularity graph. The third condition prevents the introduction of new T-junctions when collapsing chords or other invalid configurations such as a node with only two edges meeting at a corner. \Cref{fig:collapsible} shows the collapsible chords for a given quad layout.

Given these assumptions, we can define a collapse operation on a collapsible chord. We define this operation by defining two sub-operations on patches. On a collapsible chord, any patch will either have singularities on opposite corners, or it will have one or two singularities only on one longitudinal side. We refer to the former as a \emph{zip} patch, and the latter as a \emph{non-zip} patch. The green patch in \Cref{fig:t-chords} bottom is a zip patch and the other 3 are non-zip patches.

The collapse operation on a non-zip patch is to simply delete the edges on the longitudinal side without any singularities. The operation on a zip patch is to remove both longitudinal sides of the patch and replace them with a single line that connects the two singularities together. In practice, we take a weighted average of the two sides, figuratively ``zipping'' the two edges together to form the new line. If any T-junctions occur on a side that is deleted during a collapse, the hanging separatrix is simply extended after the collapse operation until it crosses the next separatrix. \Cref{fig:simplification} illustrates three consecutive chord collapses used to simplify a partition. The next chord to be collapsed in each frame has its zip patches highlighted in red and its non-zip patches highlighted in blue.

This collapse operation effectively replaces the two longitudinal sides of a chord with a single curve passing through each of the singularities on either side. It is easy to see that the resulting graph is still a T-layout because the local connectivity at singularities is not changed and other crossings of separatrices are either unaffected or simply removed (see the proof of theorem 5.4 in \cite{viertel_approach_2019}). We summarize this section with the following proposition,
\begin{figure}[h]
  \begin{center}
    \includegraphics[width=\linewidth]{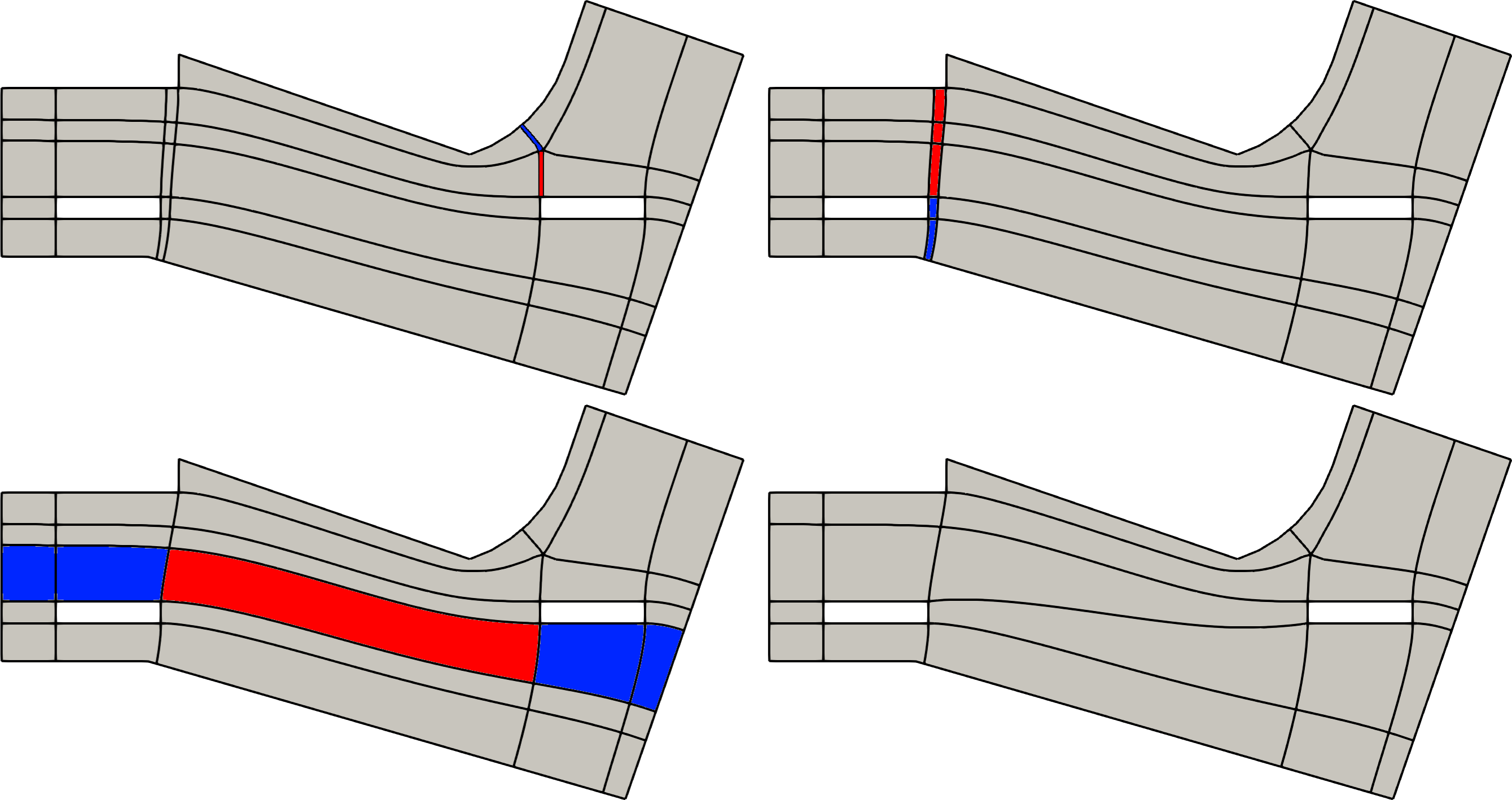}
  \end{center}
  \caption{Three consecutive chord collapses simplify the quad layout. Zip patches to be collapsed at each step are colored in red and non-zip patches are colored in blue.}
  \label{fig:simplification}
\end{figure}
\begin{proposition}
  Each chord collapse operation removes a chord from the T-layout, resulting in another T-layout with the same irregular nodes on the boundary and interior. A series of collapses monotonically reduces the number of T-junctions in the layout, and strictly decreases the number of partition components.
\end{proposition}

This simple operation forms the core of our partition simplification algorithm. As \Cref{fig:simplification,fig:spline} illustrate, repeated application of this operation has the potential to dramatically simplify a T-layout obtained from separatrix tracing. We take a greedy approach, collapsing first the thinnest chord that satisfies all conditions for collapse. The full loop is described in \cref{alg:simplification}.
\begin{figure}[h]
  \begin{center}
    \includegraphics[width=\linewidth]{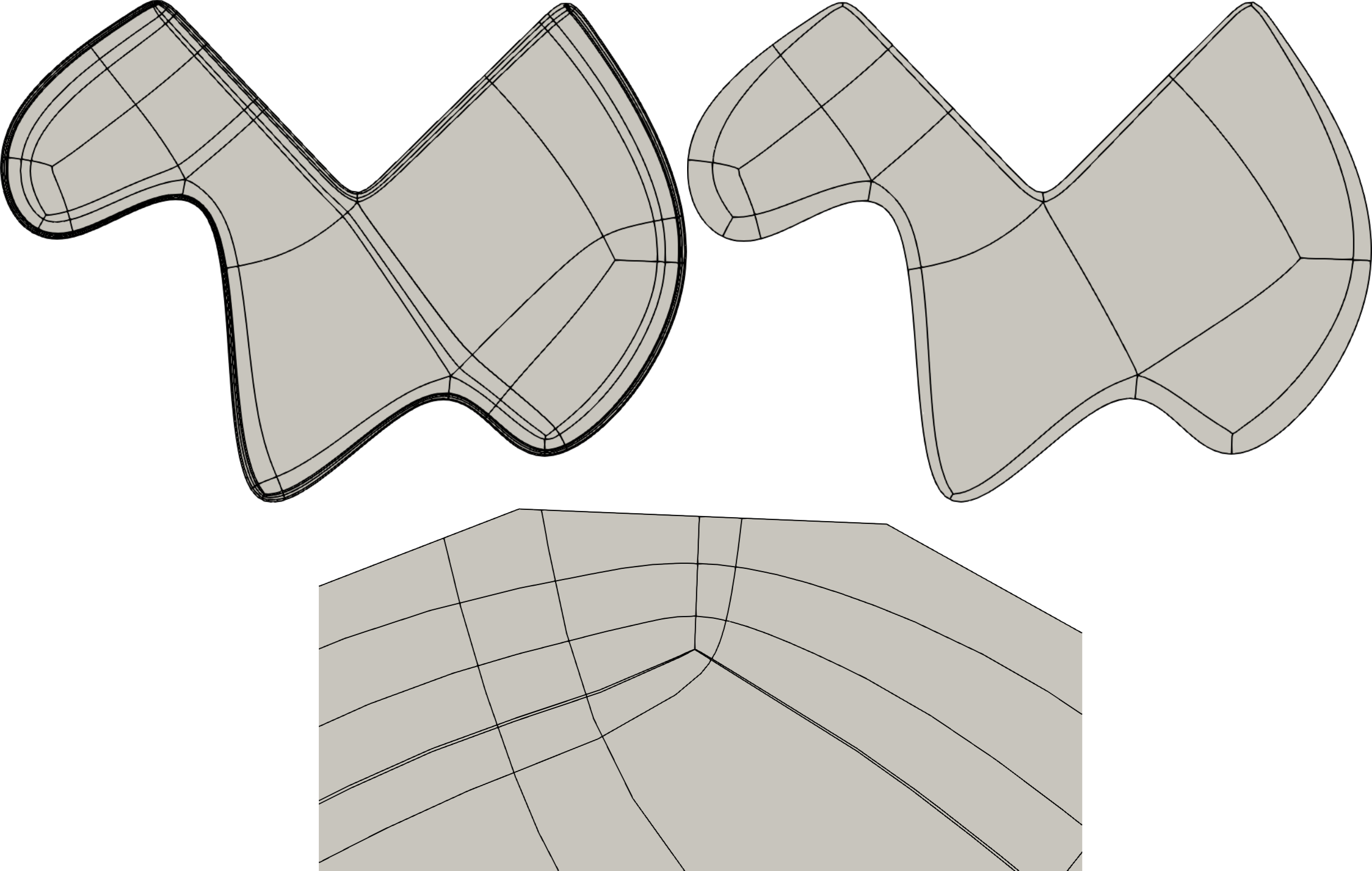}
  \end{center}
  \caption{Before and after partition simplification. {\bf Top Left:} The initial partition obtained by tracing separatrices. {\bf Top Right:} A simplified partition after 10 chord collapses. {\bf Bottom:} A close up of the top left corner of the geometry reveals extremely small components that occur because of misalignment in singularities in the initial partition.}
  \label{fig:spline}
\end{figure}
\begin{algorithm}
\caption{Partition simplification} \label{alg:simplification}
\begin{algorithmic}
\vspace{.2cm}
\STATE Let $\Gamma$ be the set of collapsible chords of the partition

\WHILE {$|\Gamma| > 0$}
\IF{No chords meet the conditions for collapse}
    \STATE $\mathbf{Stop.}$
\ELSE
    \STATE Collapse the chord with the smallest minimum width
    \STATE Determine new set of collapsible chords $\Gamma$
\ENDIF
\ENDWHILE

\end{algorithmic}
\end{algorithm}

\subsection{Conditions for Collapse}

It is not always beneficial to collapse every collapsible chord. \Cref{fig:collapsible} highlights four collapsible chords in a partition, but by most measures, it would only be beneficial to collapse the thinnest of the chords, since collapsing the others would lead to severe deformation in the newly created partition components adjacent to the zipped edge. The decision of whether to collapse is also application dependent. For example, in the final chord collapse in \Cref{fig:simplification}, the difference in length on opposite sides of the regions adjacent to the zipped separatrix may outweigh the cost of a slightly more complex partition, depending on the application.

A complete exploration of how different applications might benefit from various collapse conditions will not be treated here, rather we only describe the conditions used in our examples. We define an \emph{energy} for each patch and we subsequently define the energy of the chord as the minimum energy of any of its patches. The collapse condition evaluates to true if the energy is positive and false if the energy is negative.

For a non-zip patch, we set the energy to a positive constant value. The exact value is not important, this simply reflects the notion that collapsing a non-zip patch is not detrimental to the overall quality of the partition.

For zip patches, let $w$ be the mean of the length of each transverse rung of the patch. Let $l$ be the mean of the length of each longitudinal side of the chord. The energy for the patch is then defined as

  \begin{table*}[h!]
    \caption{Basic data for our pipeline on several models}
      \label{tab:data}

      \begin{tabular}{l@{\hskip 0.3cm}l@{\hskip 0.5cm} ccc@{\hskip 0.5cm} cc@{\hskip 0.5cm} cc@{\hskip 0.5cm} c}
        \hline\noalign{\smallskip}
        &  & Cross Field & Tracing & Simp. & \multicolumn{2}{l}{Components} & \multicolumn{2}{l}{T-junctions} & Chord \\
        Model & $n$ & (s) & (s) & (s) & Before & After & Before & After & Collapses\\
        \noalign{\smallskip}\hline\noalign{\smallskip}
        cognit & 7274 & 0.444 & 0.274 & 0.450 & 583 & 190 & 86 & 0 & 72\\
        chainr5 & 4781 & 0.105 & 0.318 & 0.494 & 440 & 176 & 88 & 0 & 61\\
        gluegun & 1842 & 0.074 & 0.124 & 0.254 & 725 & 189 & 36 & 0 & 45\\
        \hline\noalign{\smallskip}
        sprayer & 954 & 0.037 & 0.027 & 0.011 & 29 & 12 & 6 & 0 & 4\\
        faceplate & 47655 & 7.036 & 0.989 & 2.270 & 1500 & 227 & 120 & 0 & 101\\
        part29 & 3265 & 0.121 & 0.051 & 0.023 & 66 & 22 & 3 & 0 & 6\\
        \hline\noalign{\smallskip}
        test1 & 2703 & 0.115 & 0.044 & 0.027 & 47 & 19 & 1 & 0 & 5\\
        engine2 & 564 & 0.025 & 0.053 & 0.057 & 194 & 63 & 11 & 1 & 15\\
        pump & 2592 & 0.067 & 0.249 & 0.821 & 1014 & 303 & 88 & 4 & 80
      \end{tabular}
  \end{table*}
\[ e = \frac{\pi}{8} - \arctan{\frac{w}{l}}.\]
If the zip patch were perfectly rectangular, then $\arctan{\frac{w}{l}}$ would be equivalent to the angle that the diagonal makes with the base of the rectangle. In rough terms, this condition prevents chord collapses that result in a large deformation of the angles that separatrices make at singularities.

We found this particular collapse condition and the heuristic of collapsing thinnest chords first to produce quads with more rectangular corners than other collapsing strategies that we tried. \Cref{fig:energy-comparison} shows a comparison between the results of collapsing a given initial partition using the strategy that we describe versus the strategy of greedily collapsing chords via our chord collapse operation, but using an energy analogous to that used in Tarini et al.\ \cite{tarini_simple_2011} and Razafindrazaka et al.\ \cite{razafindrazaka_perfect_2015}.
\begin{figure}[h!]
  \begin{center}
    \includegraphics[width=0.8\linewidth]{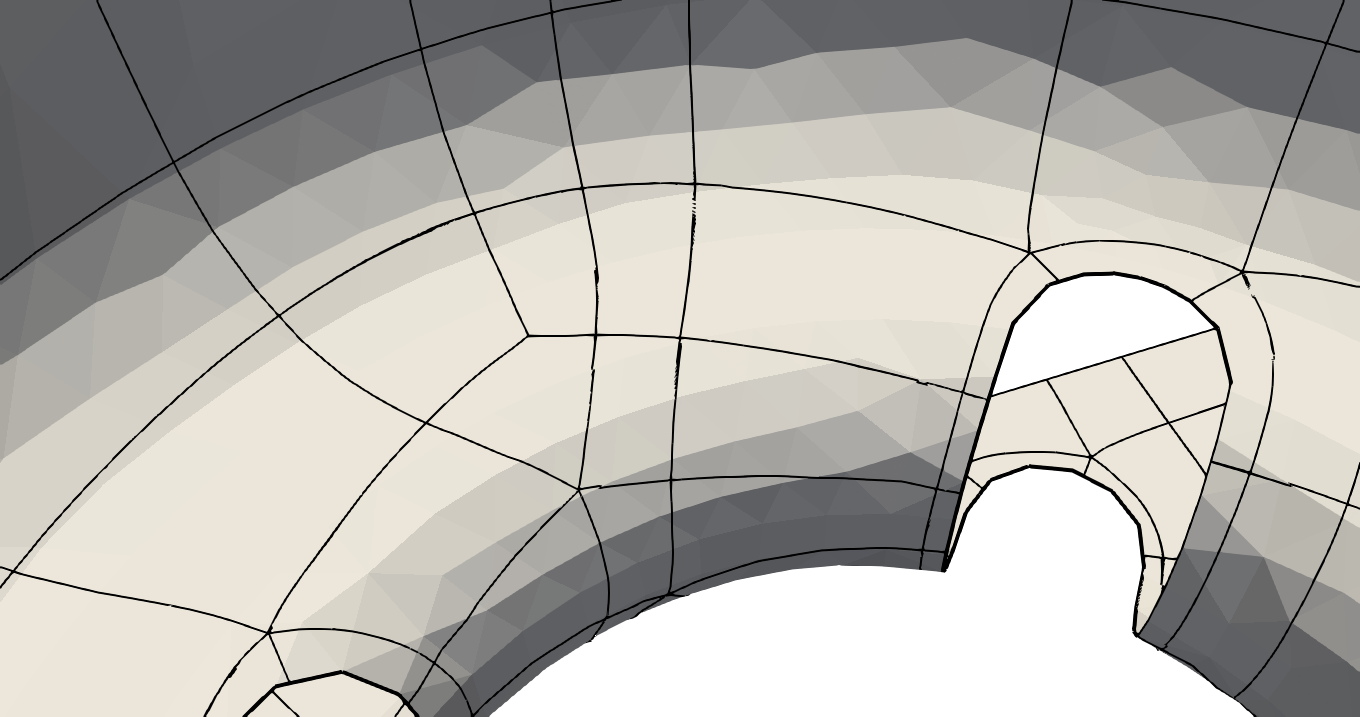} \\
    \includegraphics[width=0.8\linewidth]{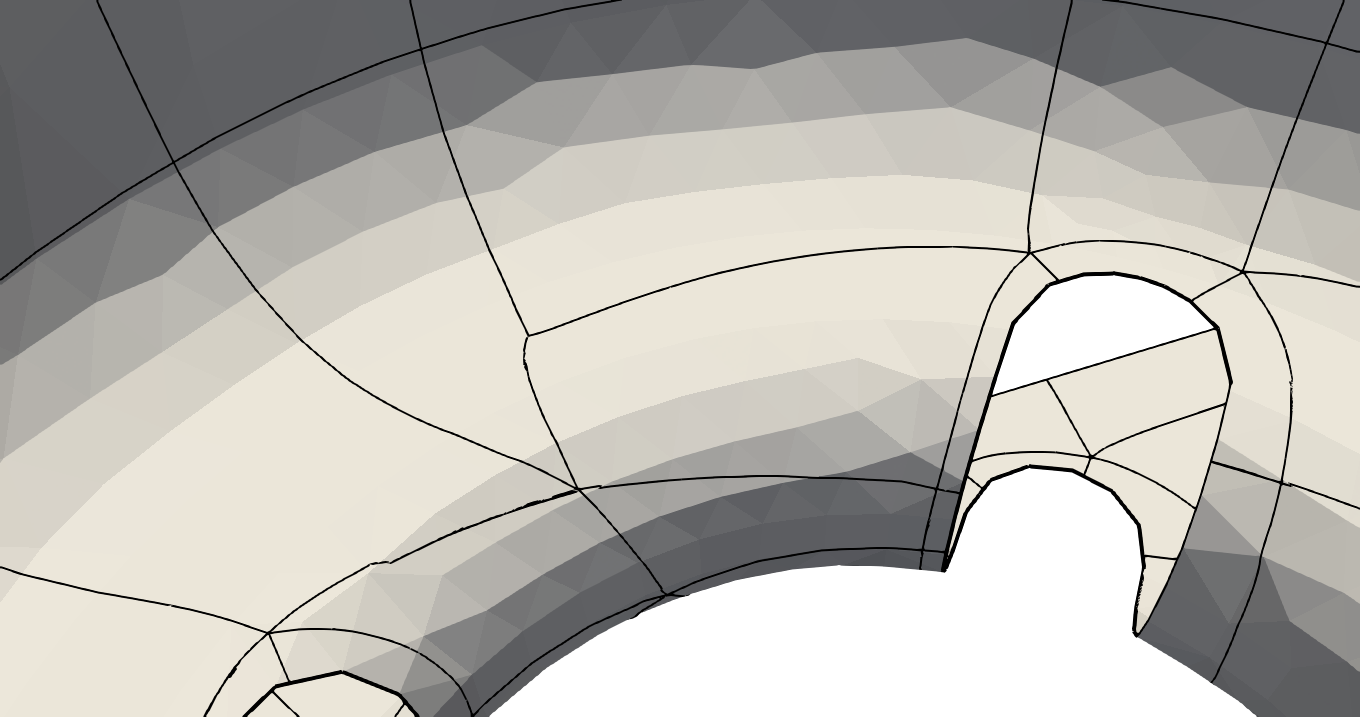}
  \end{center}
  \caption{A comparison of collapse strategies. {\bf Top:} The partition obtained by collapsing an initial partition according to the strategy defined in \cref{alg:simplification}. {\bf Bottom:} Result of collapsing the same partition using a greedy strategy collapsing chords in the order of highest to lowest energy using an energy analogous to that used in Tarini et al.\ \cite{tarini_simple_2011} and Razafindrazaka et al.\ \cite{razafindrazaka_perfect_2015}. This strategy is over-aggressive in collapsing chords and we conclude that the energy does not work well with the chord collapse approach.}
  \label{fig:energy-comparison}
\end{figure}

While we found the strategy of collapsing chords according to the conditions specified in this section to work well in our examples, it is a simple matter to substitute the sorting function and conditions for collapse in this algorithm with whatever is deemed appropriate for the application at hand.

\section{Numerical Experiments} \label{sec:results}

We tested our algorithm on 100 triangle meshes of surfaces with boundary derived from CAD models. All of the models except for the ``faceplate'' model are from a test suite used for development of the CUBIT software \cite{sandia_cubit_2017}. The ``faceplate'' model is the faceplate of the motor from the fan model at \url{https://grabcad.com/library/electric-fan-model-1}. For the diffusion generated method, we used a time step $\tau = 1/\lambda_1$ where $\lambda_1$ is the first eigenvalue of the matrix $A$. We continued the iterations until $\|\vec{u}_k - \vec{u}_{k-1}\| < \sqrt{2n} \times 10^{-6}$ where $n$ is the number of free nodes in the mesh. All examples were run on an Intel Core i5-2420m on a single thread.

In \Cref{tab:data} and \Cref{fig:examples,fig:more-examples,fig:even-more-examples}, we present nine example models that are representative of the models used and results obtained in our experiment. \Cref{tab:data} shows data for the number of nodes in the triangle mesh, timing for the diffusion generated method, streamline tracing, and partition simplification methods, the number of components and T-junctions before and after simplification, and the total number of chord collapses performed. \Cref{fig:examples,fig:more-examples,fig:even-more-examples} show the initial partition obtained via streamline tracing on the top and the final simplified partition on the bottom. The models in the table are shown in the same order as they appear in the figures, and the horizontal lines in \Cref{tab:data} identify the cutoffs between figures.

\begin{figure*}[h!]
  \includegraphics[width=0.33\linewidth]{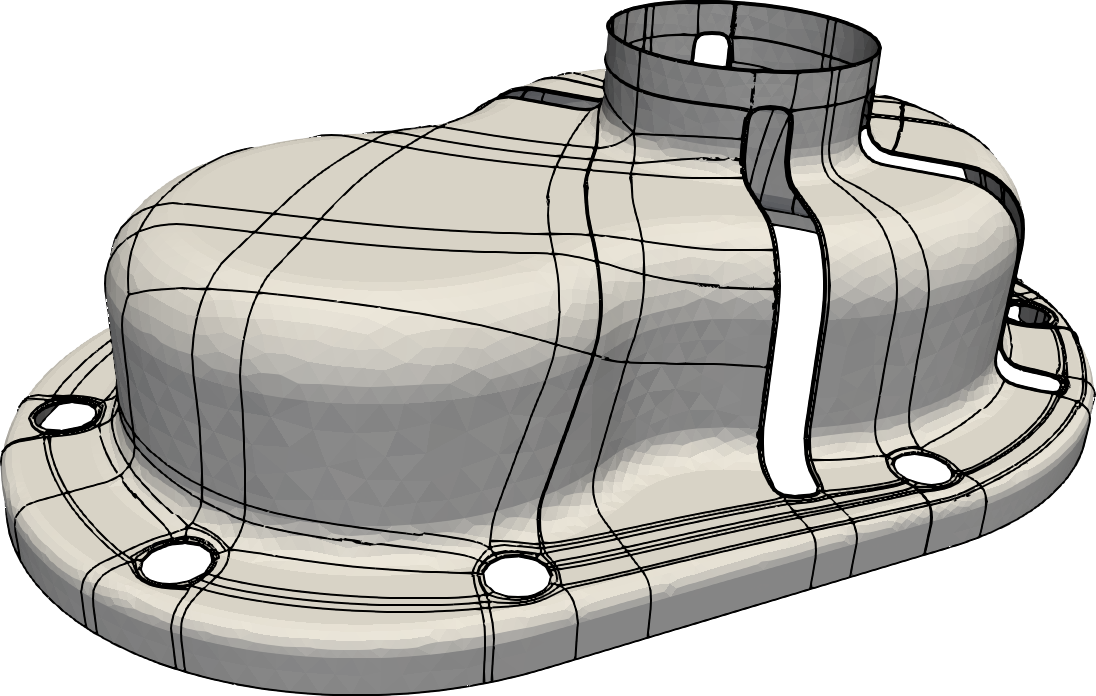}
  \includegraphics[width=0.33\linewidth]{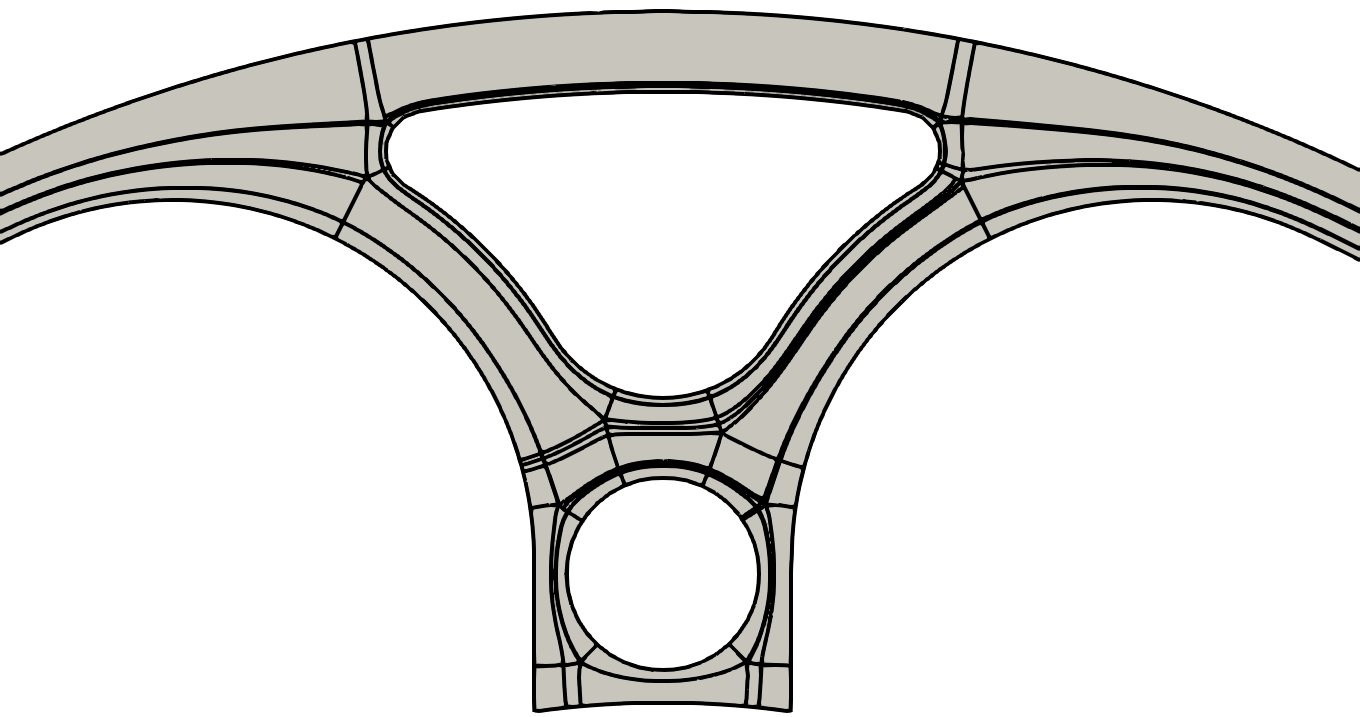}
  \includegraphics[width=0.33\linewidth]{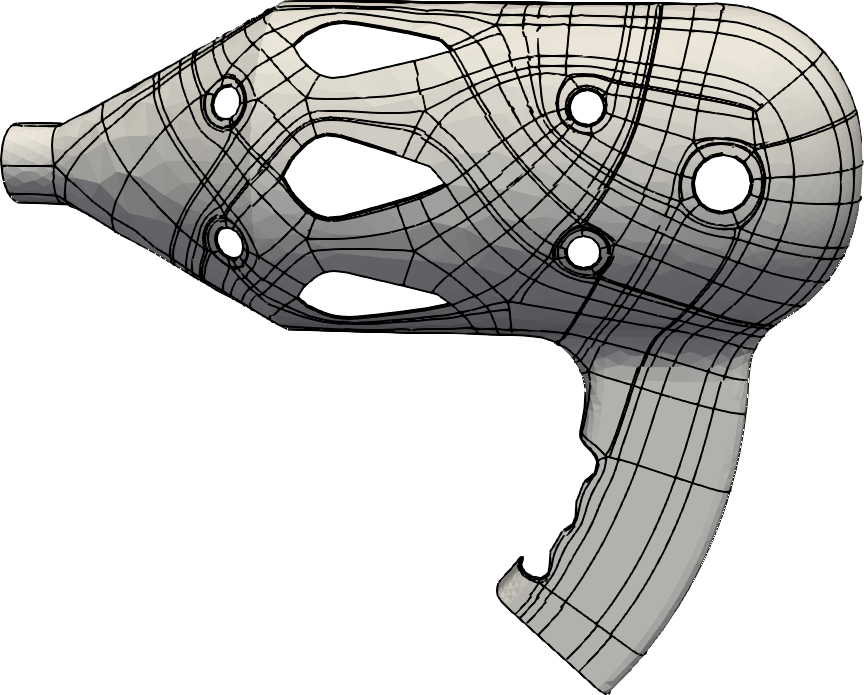}\\
  \includegraphics[width=0.33\linewidth]{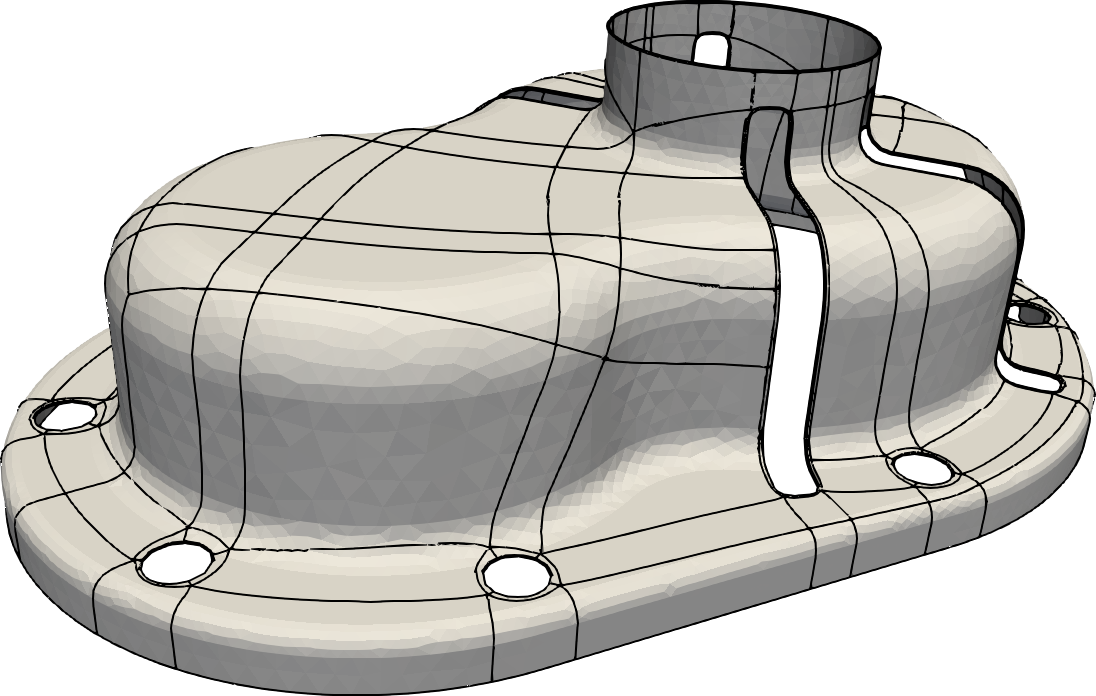}
  \includegraphics[width=0.33\linewidth]{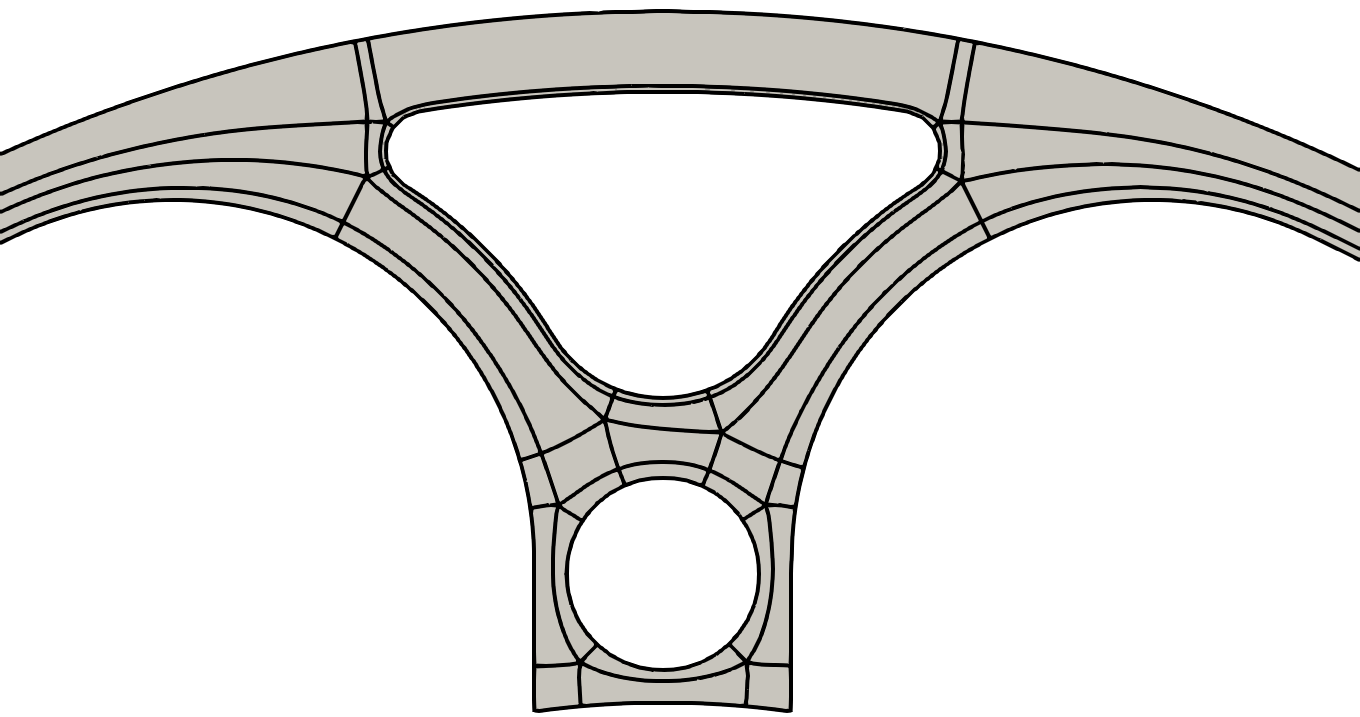}
  \includegraphics[width=0.33\linewidth]{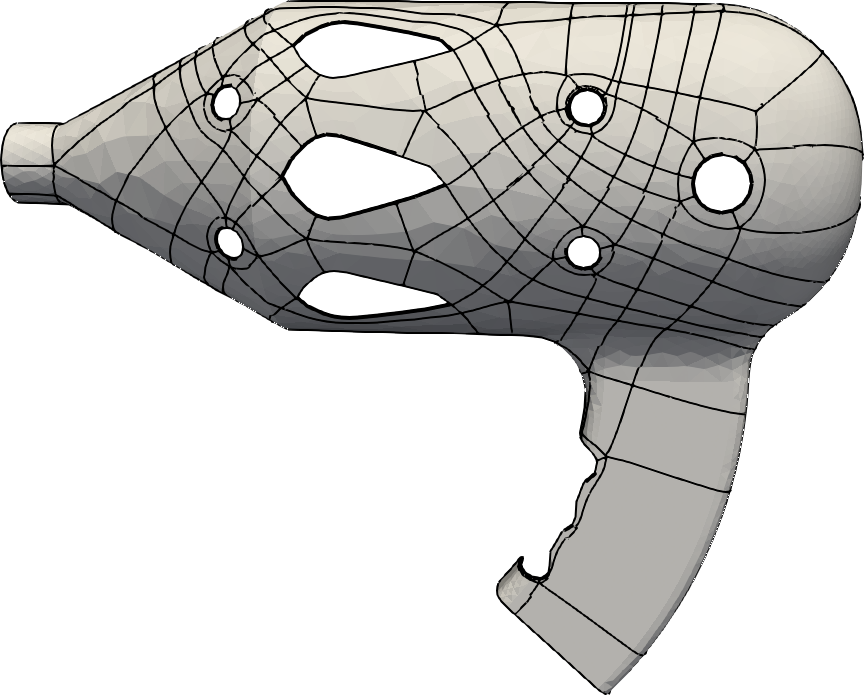}

  \caption{Examples of partitions simplified by our algorithm. The models from left to right are \emph{cognit}, \emph{chainr5}, and \emph{gluegun}. {\bf Top:} The initial partition obtained by tracing streamlines of the cross field obtained via the diffusion generated method. {\bf Bottom:} A simplification of the partition on the top via our method.}
  \label{fig:examples}
\end{figure*}
\begin{figure*}[h!]
    \includegraphics[width=0.33\linewidth]{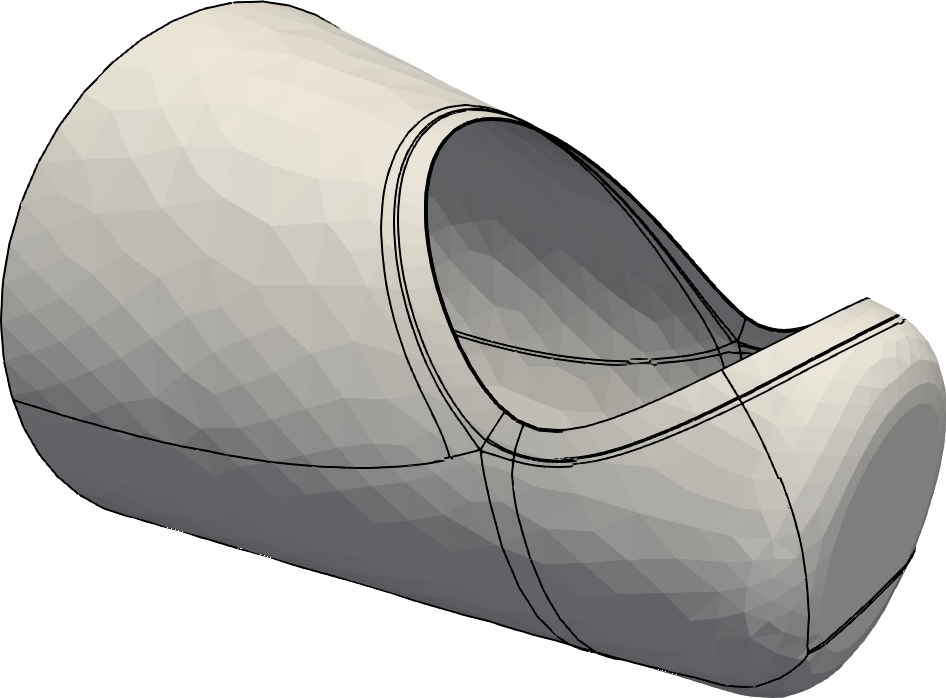}
    \includegraphics[width=0.33\linewidth]{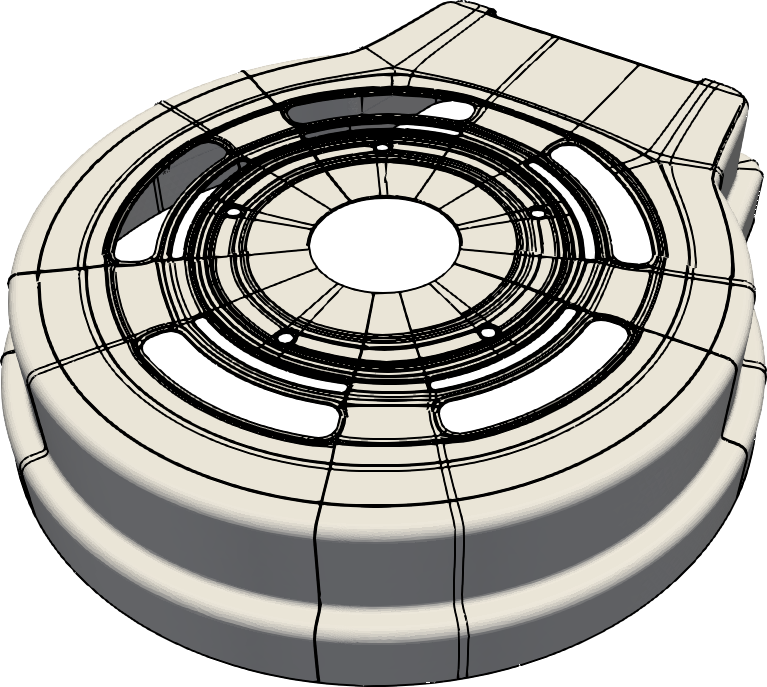}
    \includegraphics[width=0.33\linewidth]{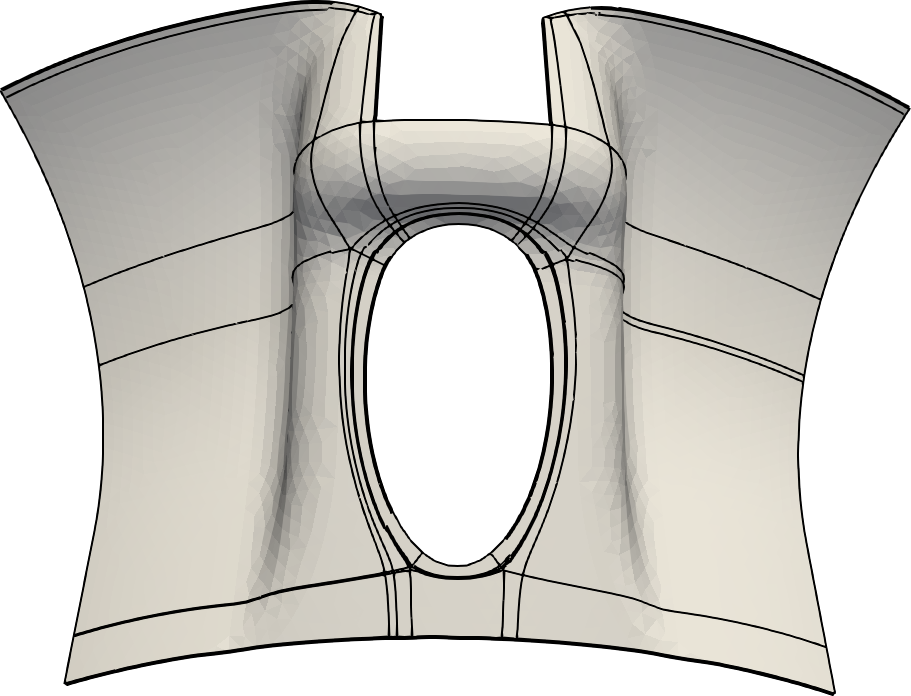}\\
    \includegraphics[width=0.33\linewidth]{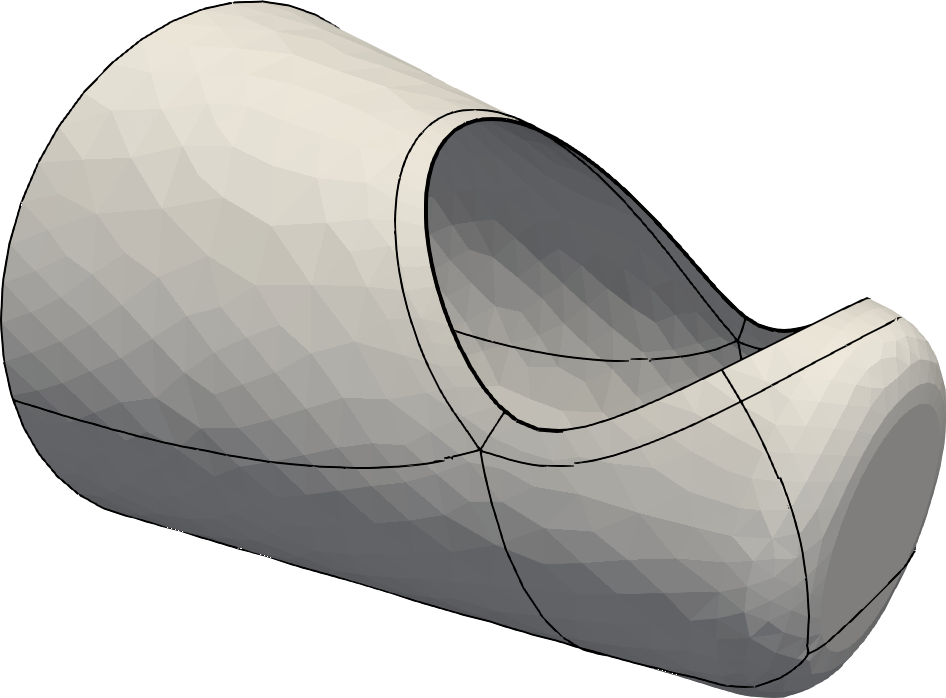}
    \includegraphics[width=0.33\linewidth]{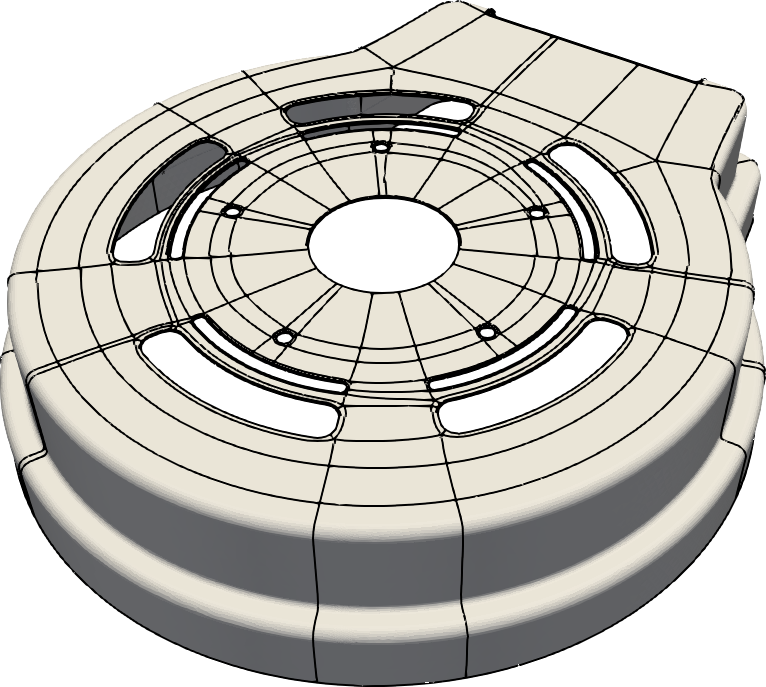}
    \includegraphics[width=0.33\linewidth]{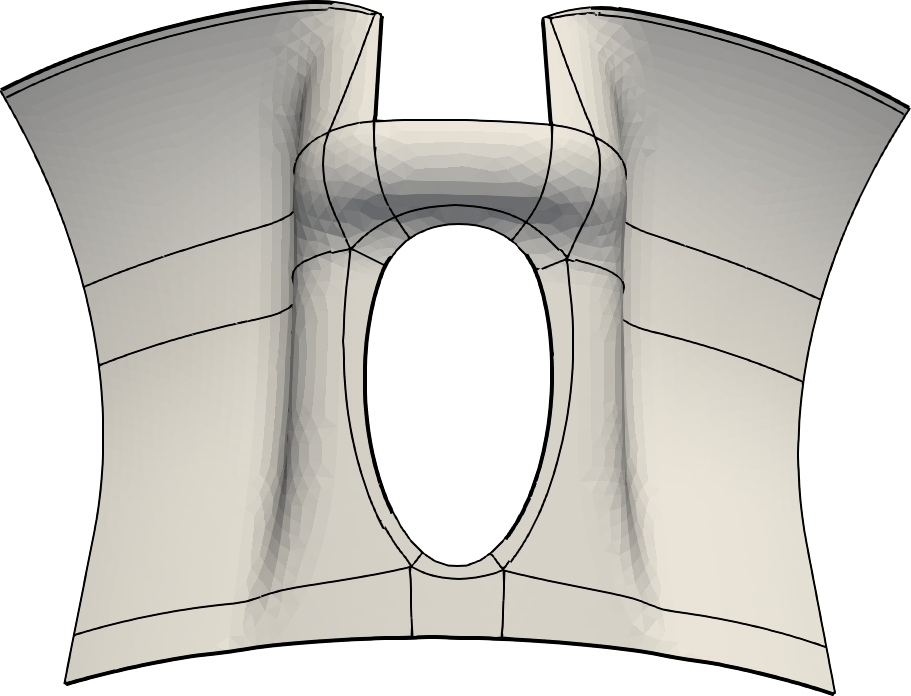}
    \caption{More examples showing the models \emph{sprayer}, \emph{faceplate}, and \emph{part29}. See the caption for \Cref{fig:examples}.}
    \label{fig:more-examples}
\end{figure*}
\begin{figure*}[h!]
  \includegraphics[width=0.33\linewidth]{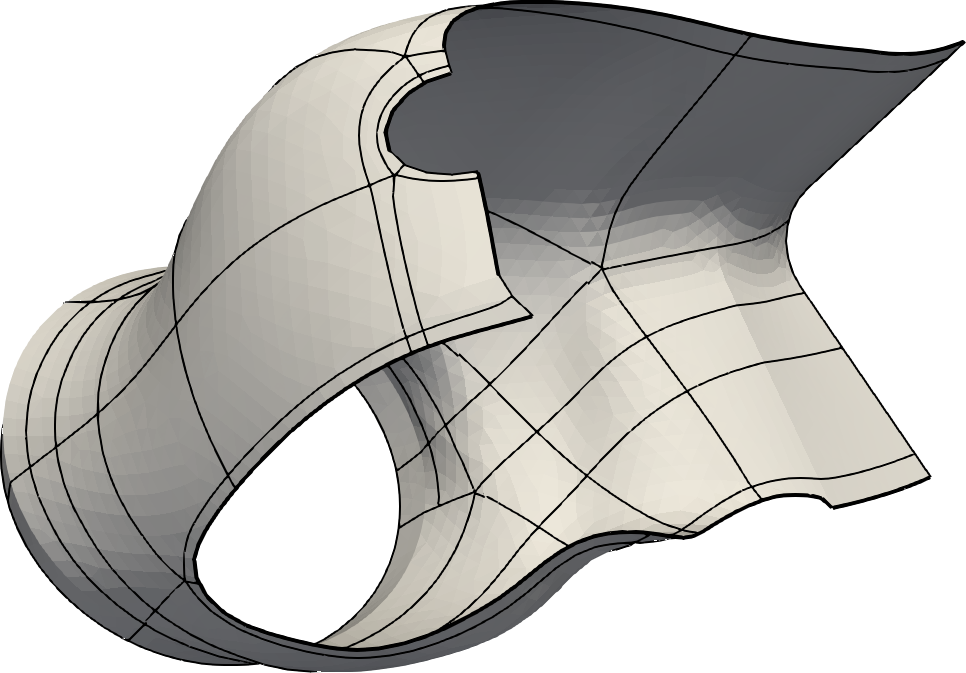}
  \includegraphics[width=0.33\linewidth]{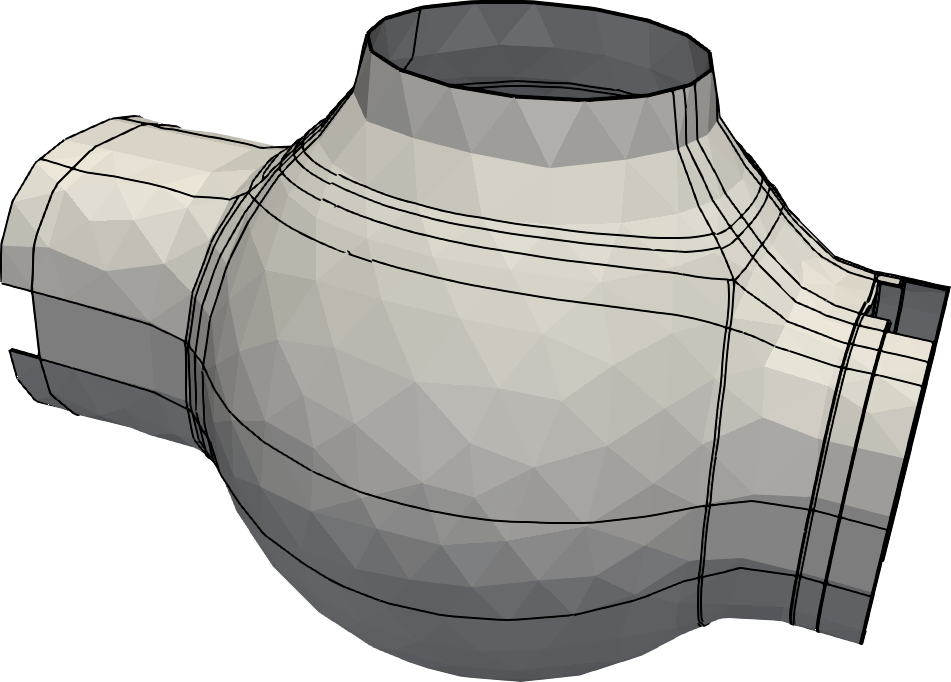}
  \includegraphics[width=0.33\linewidth]{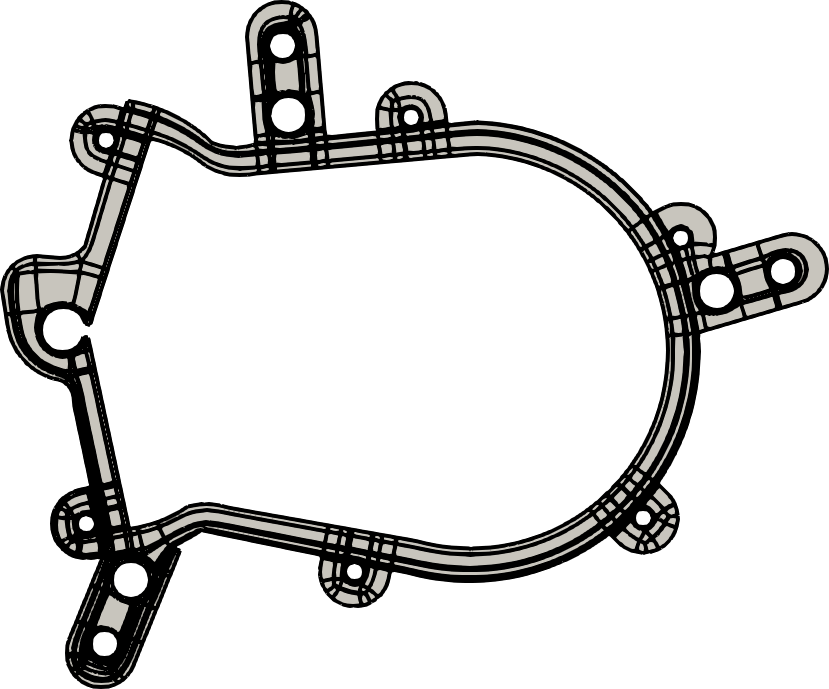}\\
  \includegraphics[width=0.33\linewidth]{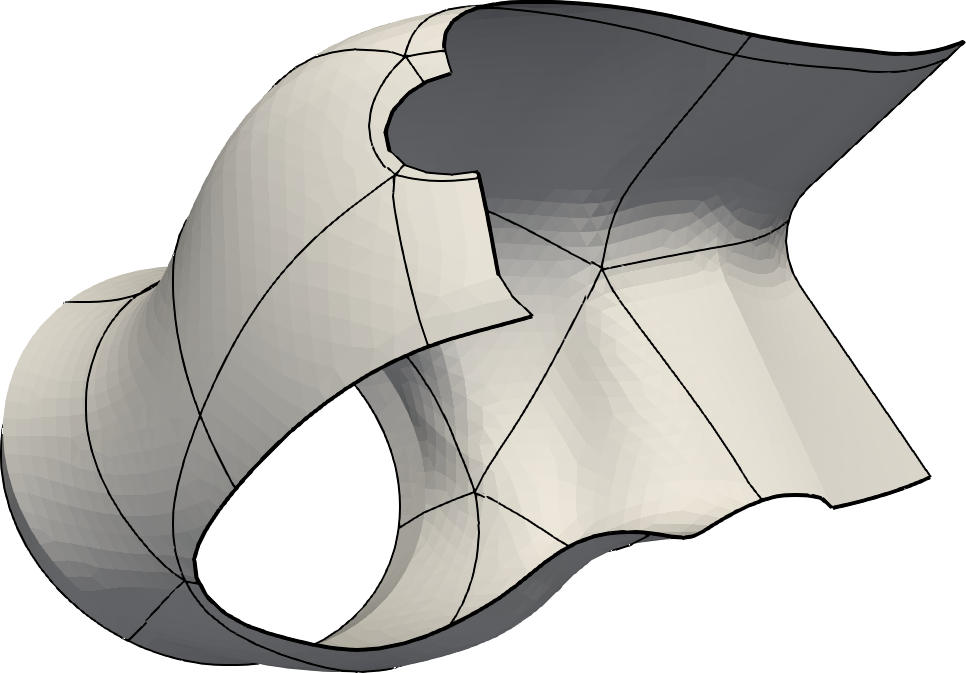}
  \includegraphics[width=0.33\linewidth]{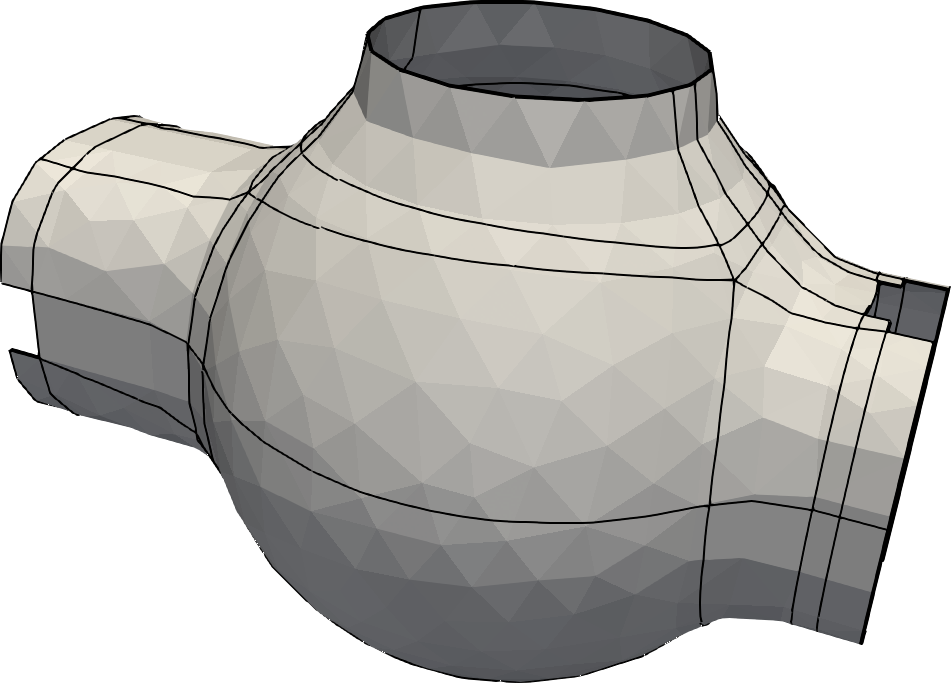}
  \includegraphics[width=0.33\linewidth]{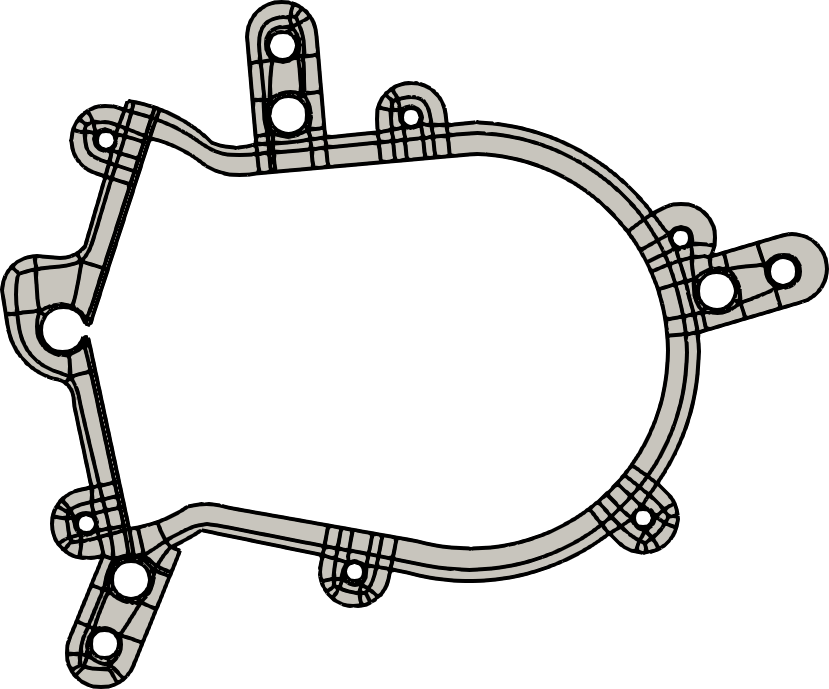}
  \caption{More examples showing the models \emph{test1}, \emph{engine2}, and \emph{pump}. See the caption for \Cref{fig:examples}.}
  \label{fig:even-more-examples}
\end{figure*}

Overall, we observe that our algorithm performs well both in terms of efficiency and results. The timings reported for our cross field design method are comparable to those for the fastest cross field design methods \cite{knoppel_globally_2013,jakob_instant_2015}. The timings for partition simplification reported in this paper are approximately an order of magnitude faster than those reported in \cite{tarini_simple_2011} and \cite{razafindrazaka_perfect_2015} on similar sized models. Visually, the coarseness of the final quad layouts appear to be comparable across all three methods; however, a better comparison using the same models with each method is needed.

Out of the database of 100 models that we tested, eight models still had T-junctions after the simplification process. On four of those models, the T-junctions could be removed by simply continuing to trace the streamlines until they reached the boundary; see \Cref{fig:removable-ts} top. On the other four, there was at least one T-junction where the corresponding streamline approached a limit cycle. In each case that we observed, all T-junctions could have been removed from the initial partition by collapsing the chords in a different order (\Cref{fig:removable-ts} bottom), which suggests that perhaps a better collapse order would prioritize or even require collapsing chords that end in T-junctions.
\begin{figure}[h!]
  \begin{center}
    \includegraphics[width=\linewidth]{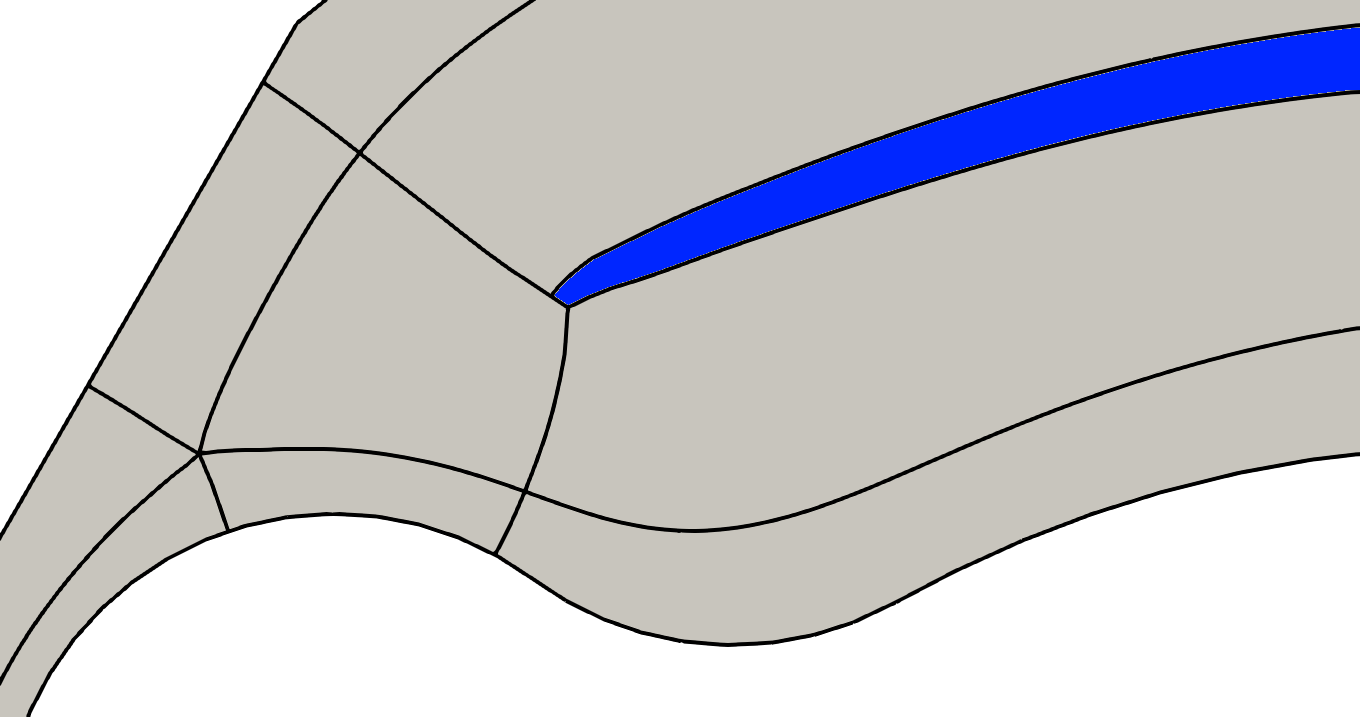}\\
    \includegraphics[width=.8\linewidth]{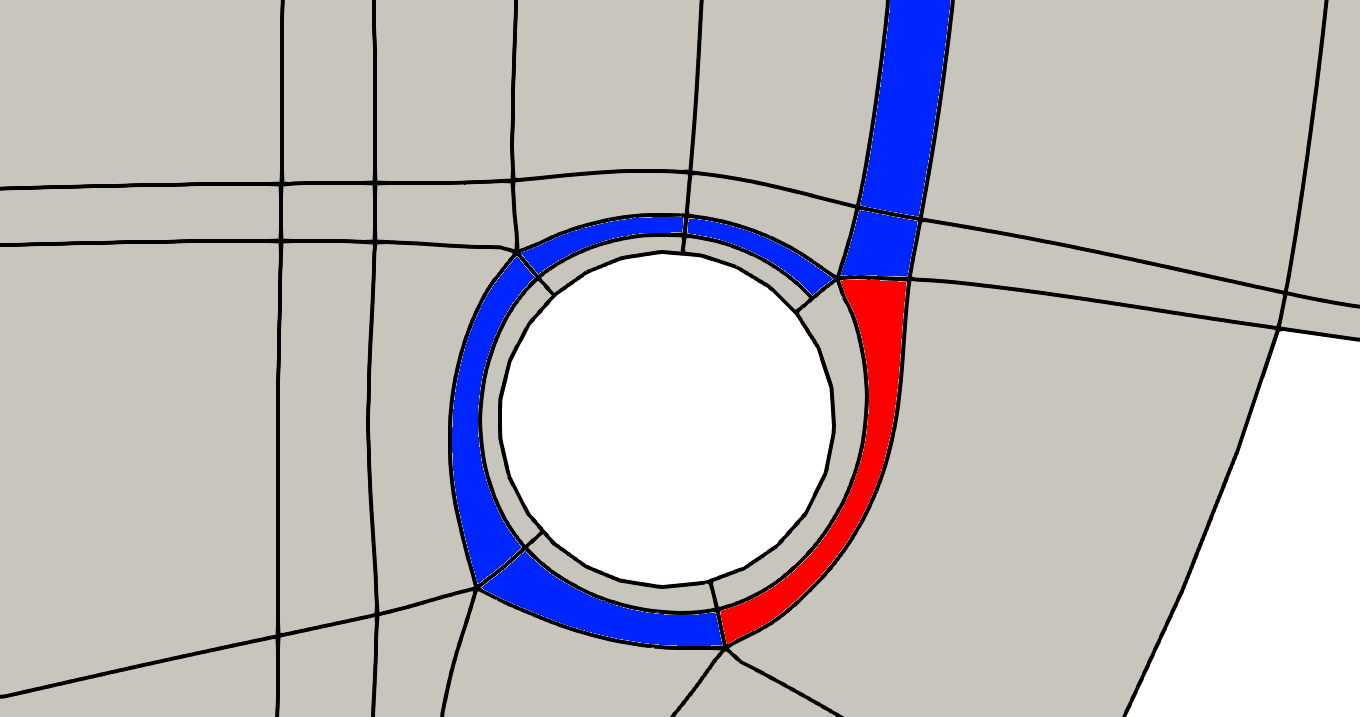}
  \end{center}
  \caption{T-junctions not eliminated by our method. {\bf Top:} A T-junction that can be removed simply by tracing the streamline until it reaches the boundary. The T-junction occurred here because the separatrix intersected another separatrix while passing through a singular triangle. Collapsing the chord (highlighted in blue, full chord not shown) would remove the T-junction, but the energy condition for collapse was not met. {\bf Bottom:} A T-junction resulting from a singularity approaching a limit cycle. Zip and non-zip patches of the chord adjacent to the T-junction are marked in red and blue respectively. The colored chord is not collapsible because further up the chord a singularity is opposite a transverse rung (not shown). In the initial partition, this T-junction could have been removed by a chord collapse without this obstruction.}
  \label{fig:removable-ts}
\end{figure}

\section{Discussion} \label{sec:discussion}

In this paper, we have further developed three parts of the pipeline described in \cref{alg:ps-overview}: an efficient method for high-quality cross field design, a method to accurately compute the trajectory of streamlines in the neighborhood of a singularity that avoids tangential crossings, and a robust partition simplification algorithm. We implemented a pipeline including these improvements, and executed our code on a database of 100 CAD surfaces. In all cases, the number of partition components and T-junctions was significantly reduced, and in 92 out of 100 cases we were able to generate a coarse quad layout with no T-junctions.

The diffusion generated method is well suited for cross field design on CAD surfaces because it results in smooth boundary aligned cross fields with good singularity placement near the boundaries. It is also comparable in speed to the fastest cross field design methods; however, a more in-depth analysis is needed to fully compare the results.

Our novel method for tracing the trajectories of streamlines near singular points is simple and allows for accurate computation while avoiding tangential crossings. Our implementation away from singularities is, however, limited by our choice of a node-based cross field representation, as we are not aware of any methods for such a representation that guarantee that streamlines in regular triangles will not cross tangentially. Our method could be improved by extending it to work in conjunction with streamline tracing methods such as Ray and Sokolov \cite{ray_robust_2014} and Myles et al.\ \cite{myles_robust_2014}.

\begin{figure}[h]
  \begin{center}
    \includegraphics[width=\linewidth]{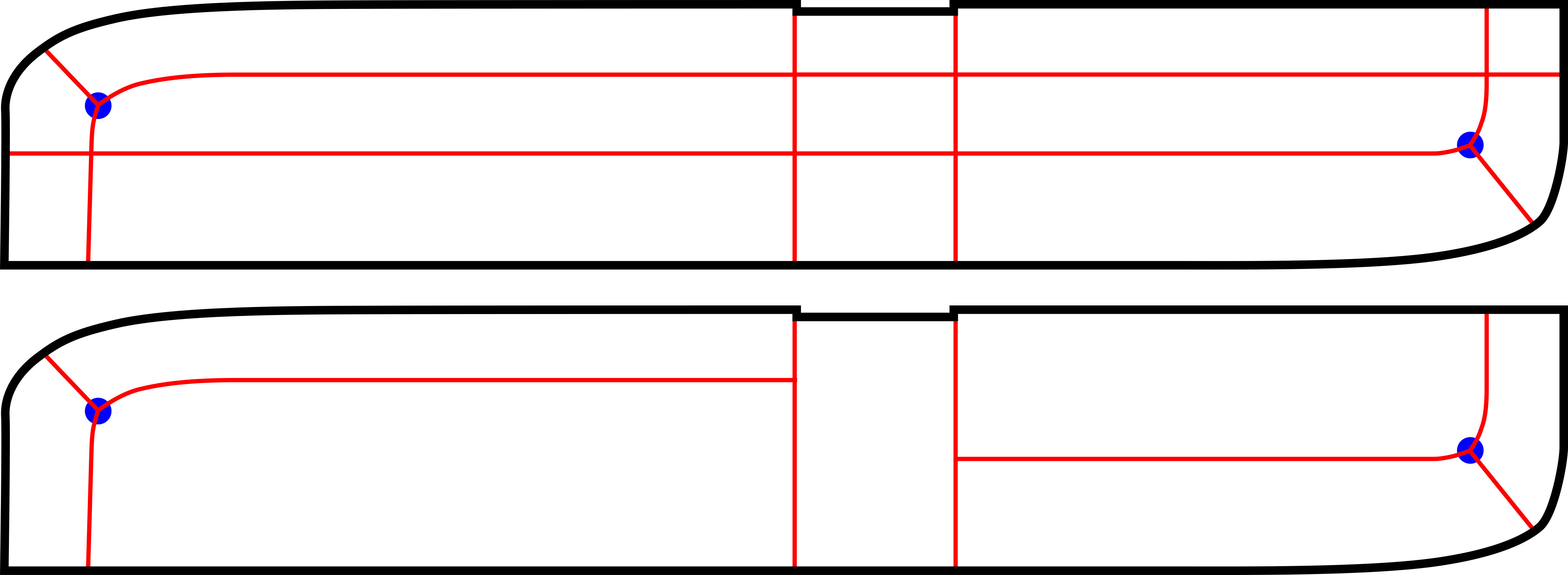}
  \end{center}
    \caption{A simple geometry where a motorcycle graph can not be simplified via a chord collapse. {\bf Top:} The partition traced out according to the conditions in section \ref{sec:tangential-crossings}. The chord running lengthwise through the center can be collapsed. {\bf Bottom:} The motorcycle graph for the same geometry. The collapsible chord from above is never formed because separatrices are cut off prior to it forming.}
    \label{fig:motorcycle-comparison}
\end{figure}

Our partition simplification algorithm is based on a simple chord collapse operation and is guaranteed to strictly decrease the number of partition components at each step as well as monotonically decrease the number of T-junctions. While our collapse operation is similar to the operation for collapsing zero-chains in Myles et al.\ \cite{myles_robust_2014}, the context in which the operation is applied is different. Perhaps the most important difference is that the T-mesh in \cite{myles_robust_2014} is a motorcycle graph, where separatrices are cut off after their first crossing with another separatrix, while in our method we trace out separatrices further, and cut them off according to conditions which allow for the collapse to have the effect of connecting two singularities together. This is illustrated in \Cref{fig:motorcycle-comparison}. Our method also collapses chords much more aggressively than \cite{myles_robust_2014}, as the primary goal of our algorithm is to generate a coarse quad layout (without T-junctions) whereas the reason for collapsing in \cite{myles_robust_2014} is to remove zero-chains which would result in a degenerate parameterization. The collapsing order in our algorithm prioritizes collapsing thin regions first whereas in \cite{myles_robust_2014} there is no discussion of order. There are also some subtle differences between the definitions of the operations themselves. For example, in \cite{myles_robust_2014}, the definition of collapsible zero-chains depends on the assignment of parametric lengths to edges of the input T-layout, whereas our definition of a chord is strictly geometric. Further, the notion of a \emph{patch} in our operation allows for zip operations spanning multiple quads where a zip like operation in \cite{myles_robust_2014} always occurs across a single quad. The cumulative effect of these differences is that we are able to demonstrate that the iterative collapse of chords can be an effective tool for generating coarse quad layouts, many times eliminating all T-junctions, whereas in \cite{myles_robust_2014}, the collapse operation is used in a limited scope, with the purpose of ensuring global consistency of their parametric length assignment.

\begin{figure*}[h]
  \begin{center}
    \includegraphics[width=0.3\linewidth]{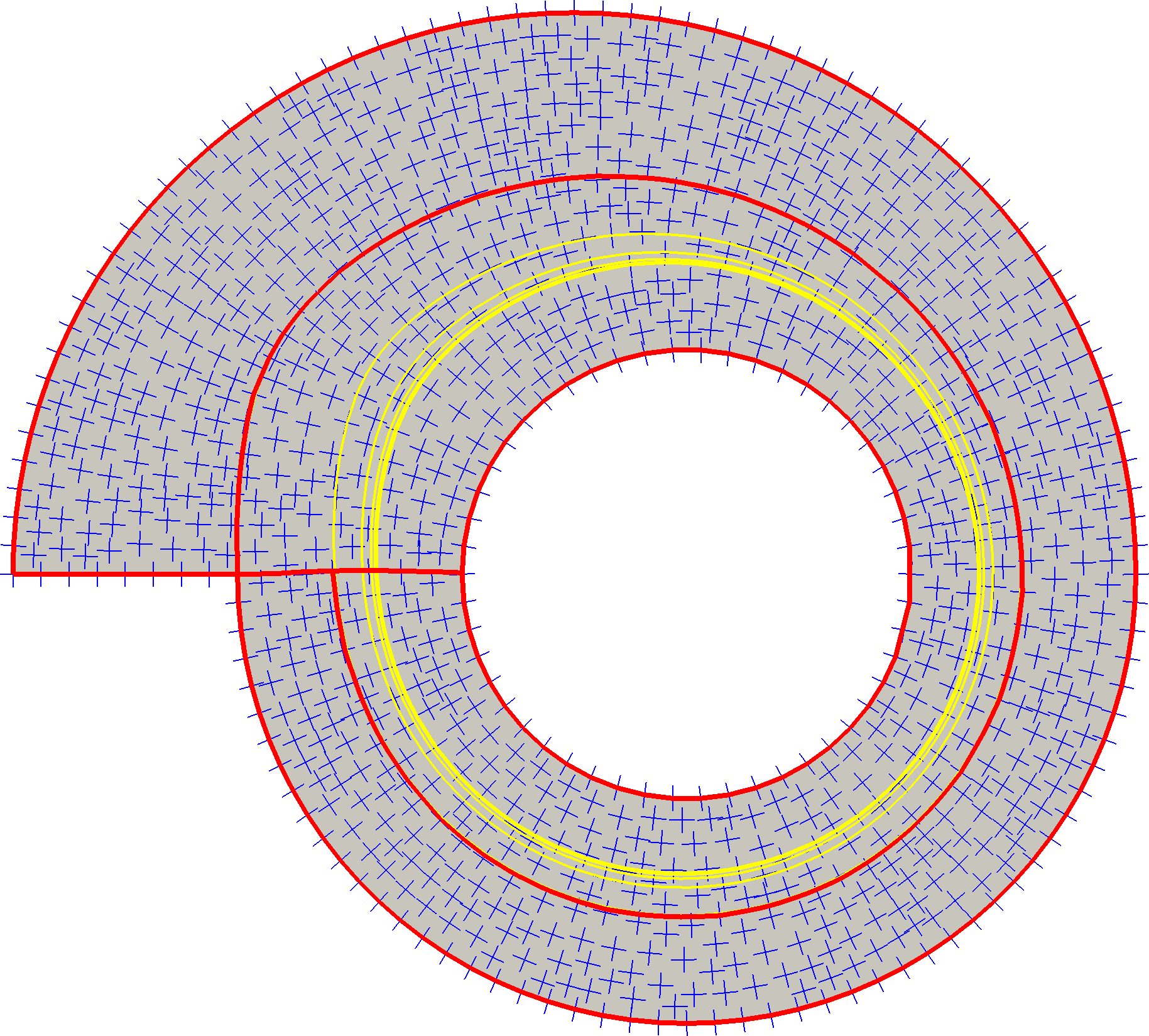}\hspace{5em}
    \includegraphics[width=0.3\linewidth]{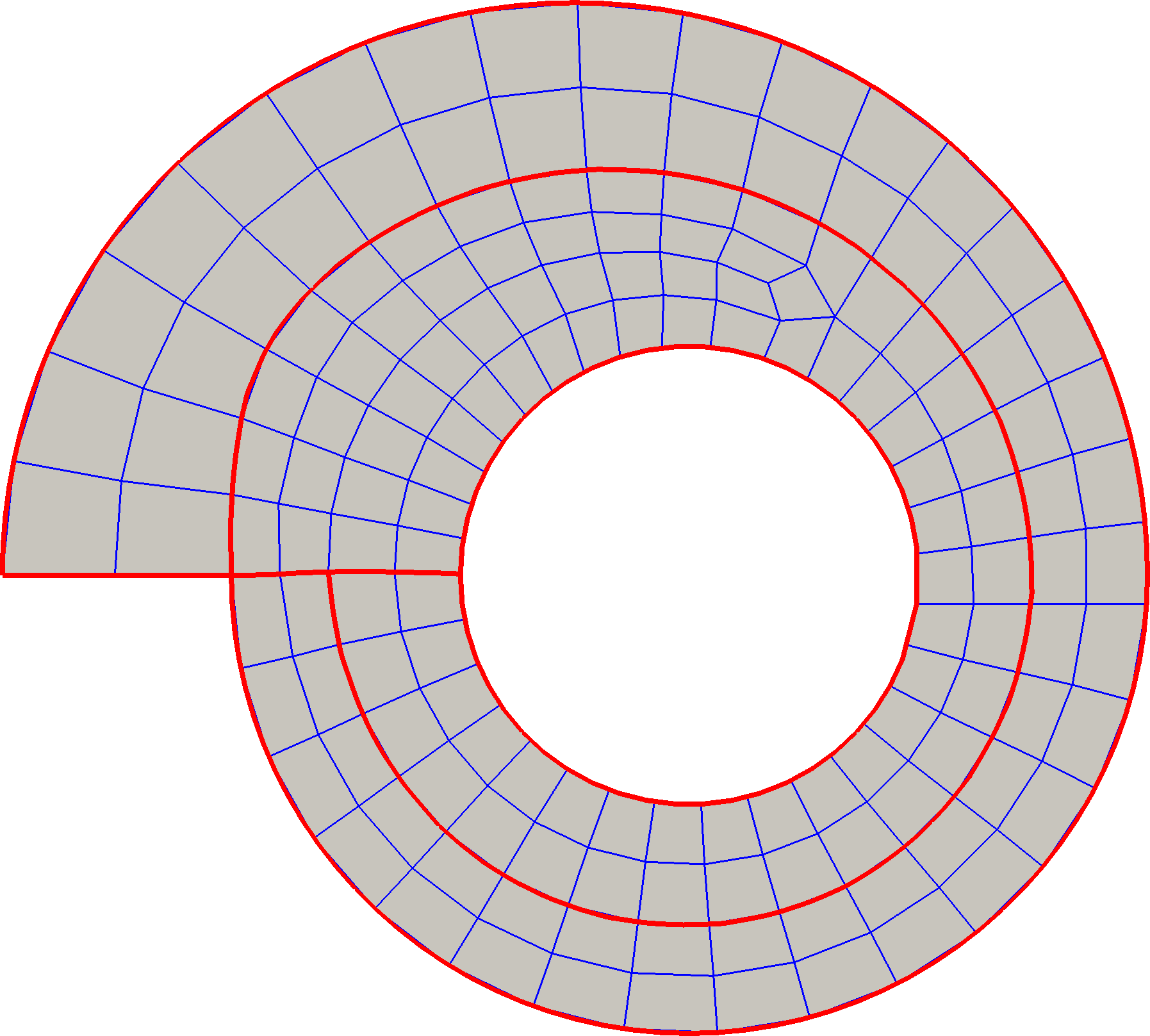}
  \end{center}
    \caption{A surface where the cross field contains a limit cycle, resulting in a T-junction that cannot be removed without the introduction of additional singularities. {\bf Left:} A boundary aligned cross field is shown in blue. The T-layout obtained by tracing out streamlines is shown in red. The yellow streamline begins at the geometric corner, follows the red curve of the partition, and continues on to converge to a limit cycle. {\bf Right:} Because of the limit cycle, a pair of 3- and 5-valent nodes is needed to mesh the region adjacent to the T-junction.}
    \label{fig:nautilus}
\end{figure*}
Our collapse method works directly on streamlines, so it does not require pre-meshing like methods such as Tarini et al.\ \cite{tarini_simple_2011} and Bommes et al.\ \cite{bommes_global_2011}, or prior computation of a seamless parameterization like Razafindrazaka et al.\ \cite{razafindrazaka_perfect_2015}, Campen et al.\ \cite{campen_quantized_2015}, or Zhang et al.\ \cite{zhang_automatic_2016}. Since our method requires tracing streamlines accurately enough to avoid tangential crossings, the number of streamlines to be traced in our method scales linearly with the number of singularities as opposed to the method of Pietroni et al.\ \cite{pietroni_tracing_2016}, who allow computation of streamlines that cross tangentially at the cost that the number of streamlines to be drawn increases with combinatorial complexity as the number of singularities increases.

Each of \cite{razafindrazaka_perfect_2015,campen_quantized_2015,zhang_automatic_2016,pietroni_tracing_2016} formulate their problems as binary optimization problems, however, in order to achieve reasonable run times, they each significantly prune the search space by employing some clever method which leverages the structure of the T-layout generated by the separatrices of the underlying cross field or parameterization. While it appears that our method is at a disadvantage because it takes a greedy approach, it is not clear to what degree the final results of each method are driven by heuristics or optimization. In addition, its not clear how well indirect objectives such as the total length of separatrices weighted by how far they drift from the underlying field, reflect objectives such as minimizing the number of quad components, or maximizing the minimal width of a chord which can be pursued directly via our method. A specific application with objectively stated goals and a common set of models is needed for a clear comparison between the quality of the quad layouts generated by all of the methods mentioned above.

% In contrast to our method, the method of Razafindrazaka et al.\ \cite{razafindrazaka_perfect_2015} only considers port matchings that result in an all quad mesh, and in the cases where our method failed to remove T-junctions, would have succeeded. The main downside to their approach is that there is no way to guarantee that any solution exits to the perfect matching problem. Indeed, on the surface shown in \Cref{fig:nautilus}, there is no all quad layout for the set of singularities defined by the cross field. Where their method would fail, ours still produces a quad layout with a T-junction, which can still easily be meshed with a pattern-based technique.

The main downside of our partition simplification method is that it does not completely eliminate T-junctions from the layout. It is not always possible to produce a quad layout with no T-junctions for a given set of singularities, so an important direction for future work is to develop a method that uses strategic insertion or removal of singularities in order to guarantee that the elimination all T-junctions from a given input T-layout is possible. Such a result would be beneficial for applications in meshing for FEM, surface reconstruction into CAD from image data or geometries generated via topology optimization, and in constructing spline bases for isogeometric analysis.

\section*{Acknowledgments}
B. Osting is partially supported by NSF DMS 16-19755 and 17-52202. R. Viertel and M. Staten are supported by Sandia National Laboratories. Sandia National Laboratories is a multi-mission laboratory managed and operated by National Technology \& Engineering Solutions of Sandia, LLC., a wholly owned subsidiary of Honeywell International, Inc., for the U.S. Department of Energy's National Nuclear Security Administration under contract DE-NA0003525. SAND2019-5668 C.

\bibliography{partition-simplification}

\end{document}